\documentclass[english,runningheads,10pt]{llncs}

\usepackage{tabularx,booktabs,multirow,delarray,array}
\usepackage{graphicx,amssymb,amsmath,amssymb}
\usepackage{enumerate}
\usepackage{algorithmic}
\usepackage[ruled]{algorithm2e}
\usepackage{wrapfig}
\usepackage{fullpage}
\usepackage{latexsym}
\usepackage{lineno}

\newenvironment{proof}{\par\noindent{\bf Proof:}}{\mbox{}\hfill$\qed$\\}

\newcommand{\ignore}[1]{ }

\newcounter{rem}
\def\qed{\hbox{\rlap{$\sqcap$}$\sqcup$}}

\begin{document}

\title{Fault-Tolerant Additive Weighted Geometric Spanners}

\author{
Sukanya Bhattacharjee\inst{1}
\and
R. Inkulu\inst{1}\thanks{This research is supported in part by NBHM grant 248(17)2014-R\&D-II/1049}
}

\institute{
Department of Computer Science \& Engineering\\
IIT Guwahati, India\\
\email{\{bsukanya,rinkulu\}@iitg.ac.in}
}

\maketitle

\pagenumbering{arabic}
\setcounter{page}{1}

\begin{abstract}
Let $S$ be a set of $n$ points and let $w$ be a function that assigns non-negative weights to points in $S$. 
The additive weighted distance $d_w(p, q)$ between two points $p,q \in S$ is defined as $w(p) + d(p, q) + w(q)$ if $p \ne q$ and it is zero if $p = q$.
Here, $d(p, q)$ is the geodesic Euclidean distance between $p$ and $q$.
For a real number $t > 1$, a graph $G(S, E)$ is called a {\it $t$-spanner} for the weighted set $S$ of points if for any two points $p$ and $q$ in $S$ the distance between $p$ and $q$ in graph $G$ is at most $t$.$d_w(p, q)$ for a real number $t > 1$.
For some integer $k \geq 1$, a $t$-spanner $G$ for the set $S$ is a {\it $(k, t)$-vertex fault-tolerant additive weighted spanner}, denoted with $(k, t)$-VFTAWS, if for any set $S' \subset S$ with cardinality at most $k$, the graph $G \setminus S'$ is a $t$-spanner for the points in $S \setminus S'$.
For any given real number $\epsilon > 0$, we present algorithms to compute a $(k, 4+\epsilon)$-VFTAWS for the metric space $(S, d_w)$ resulting from the points in $S$ belonging to any of the following: $\mathbb{R}^d$, simple polygon, polygonal domain, and terrain.
Note that $d(p, q)$ is the geodesic Euclidean distance between $p$ and $q$ in the case of simple polygons and terrains whereas in the case of $\mathbb{R}^d$ it is the Euclidean distance along the line segment joining $p$ and $q$.
\end{abstract}

\begin{keywords}
Computational Geometry, Geometric Spanners, Approximation Algorithms
\end{keywords}

\section{Introduction}
\label{sect:intro}

When designing geometric networks on a given set of points in a metric space, it is desirable for the network to have short paths between any pair of nodes while being sparse with respect to the number of edges. 
Let $G(S, E)$ be an edge-weighted geometric graph on a set $S$ of $n$ points in $\mathbb{R}^d$.
The weight of any edge $(p, q) \in E$ is the Euclidean distance $|pq|$ between $p$ and $q$.
The distance in $G$ between any two nodes $p$ and $q$, denoted by $d_G(p, q)$, is defined as the length of a shortest (that is, minimum-weighted) path between $p$ and $q$ in $G$.
The graph $G$ is called a {\it $t$-spanner} for some $t \ge 1$ if for any two points $p, q \in S$ we have $d_G(p, q) \leq t.|pq|$.
The smallest $t$ for which $G$ is a $t$-spanner is called the {\it stretch factor} of $G$, and the number of edges of $G$ is called its size.

Peleg and Sch\"{a}ffer~\cite{journals/jgt/PelegSchaffer89} introduced spanners in the context of distributed computing and by Chew~\cite{journals/jcss/Chew89} in a geometric context.
Alth\"{o}fer et~al. \cite{journals/dcg/AlthoferDDJS93} first attempted to study sparse spanners on edge-weighted graphs that have the triangle-inequality property.
The text by Narasimhan and Smid \cite{books/compgeom/narsmid2007}, handbook chapter \cite{hb/cg/Epp99}, and Gudmundsson and Knauer~\cite{hb/apprxheu/GudKnau07} detail various results on Euclidean spanners, including a $(1+\epsilon)$-spanner for the set $S$ of $n$ points in $\mathbb{R}^d$ that has $O(\frac{n}{\epsilon^{d-1}})$ edges for any $\epsilon > 0$.

Many variations of sparse spanners have been studied, including spanners of low degree \cite{journals/cgta/ABCGHSV08,conf/cccg/CarmiChai10,journals/cgta/BCCCKL13,conf/cats/Smid06}, spanners of low weight \cite{journals/algorithmica/BCFMS10,journals/ijcga/DasNara97,journals/siamjc/GudmudLevcoNar02}, spanners of low diameter \cite{conf/focs/AryaMS94,journals/cgta/AryaMountSmid99}, planar spanners \cite{conf/esa/ArikatiCCDSZ96,journals/jcss/Chew89,conf/optalgo/DasJoseph89,conf/wads/KeilGutwin89}, spanners of low chromatic number \cite{journals/cgta/BCCMSZ09}, fault-tolerant spanners \cite{journals/dcg/ABFG09,journals/dcg/CzumajZ04,journals/algorithmica/LevcopoulosNS02,conf/wads/Lukovszki99,journals/corr/KapLi13,conf/stoc/Solomon14}, low power spanners \cite{journals/wireless/Karim11,conf/infocom/SegalShpu10,journals/jco/WangLi06}, kinetic spanners \cite{journals/dcg/AbamBerg11,journals/cgta/ABG10}, angle-constrained spanners \cite{journals/jocg/CarmiSmid12}, and combinations of these \cite{conf/stoc/AryaDMSS95,journals/algorithmica/AryaS97,journals/algorithmica/BFRV18,journals/algorithmica/BoseGudSmid05,journals/ijcga/BoseSmidXu09,journals/corr/CarmiChait10}.
When the doubling dimension of a metric space is bounded, results applicable to the Euclidean settings are given in \cite{conf/stoc/Talwar04}.

As mentioned in Abam et~al., \cite{journals/algorithmica/AbamBFGS11}, the cost of traversing a path in a network is not only determined by the lengths of the edges on the path but also by the delays occurring at the nodes on the path.
The result in \cite{journals/algorithmica/AbamBFGS11} models these delays with the additive weighted metric.
Let $S$ be a set of $n$ points in $\mathbb{R}^d$. 
For every $p \in S$, let $w(p)$ be the non-negative weight associated to $p$. 
The following additive weighted distance function $d_w$ on $S$ defining the metric space $(S, d_w)$ is considered in \cite{journals/algorithmica/AbamBFGS11}: 
for any $p, q \in S$, $d_w(p, q)$ equals to $0$ if $p = q$; otherwise, it is equal to $w(p) + |pq| + w(q)$.
\ignore {
\[
d_w(p, q) =
\begin{cases}
    0 & \text{if } p = q. \\
    w(p) + |pq| + w(q) & \text{if } p \ne q.
\end{cases}
\]
}

Recently, Abam et~al. \cite{conf/soda/AbamBS17} showed that there exists a $(2 + \epsilon)$-spanner with a linear number of edges for the metric space $(S, d_w)$ that has bounded doubling dimension in \cite{journals/siamcomp/Har-PeledM06}.
And, \cite{journals/algorithmica/AbamBFGS11} gives a lower bound on the stretch factor, showing that $(2+\epsilon)$ stretch is nearly optimal.
Bose et~al. \cite{conf/swat/BoseCC08} studied the problem of computing spanners for a weighted set of points.
They considered the points that lie on the plane to have positive weights associated to them; and defined the distance $d_w$ between any two distinct points $p, q \in S$ as $d(p,q) - w(p) - w(q)$. 
Under the assumption that the distance between any pair of points is non-negative, they showed the existence of a $(1 + \epsilon)$-spanner with $O(\frac{n}{\epsilon})$ edges. 

A simple polygon $P_\mathcal{D}$ containing $h \ge 0$ number of disjoint simple polygonal holes within it is termed the {\it polygonal domain} $\mathcal{D}$.
(When $h$ equals to $0$, the polygonal domain $\mathcal{D}$ is a simple polygon.)
The free space $\mathcal{F(D)}$ of the given polygonal domain ${\cal D}$ is defined as the closure of $P_\mathcal{D}$ excluding the union of the interior of polygons contained in $P_\mathcal{D}$.
Essentially, a path between any two given points in $\mathcal{F(D)}$ needs to be in the free space $\mathcal{F(D)}$ of $\mathcal{D}$. 
Given a set $S$ of $n$ points in the free space $\mathcal{F(D)}$ defined by the polygonal domain $\cal D$, computing geodesic spanners in $\mathcal{F(D)}$ is considered in Abam et~al. \cite{conf/compgeom/AbamAHA15}.
For any two distinct points $p, q \in S$, $d_\pi(p, q)$ is defined as the geodesic Euclidean distance along a shortest path $\pi(p, q)$ between $p$ and $q$ in $\mathcal{F(D)}$.
\cite{conf/compgeom/AbamAHA15} showed that for the metric space $(S, \pi)$, for any constant $\epsilon > 0$, there exists a $(5+\epsilon)$-spanner of size $O(\sqrt{h}n (\lg n)^{2})$. 
Further, for any constant $\epsilon > 0$, \cite{conf/compgeom/AbamAHA15} gave a $(\sqrt{10}+\epsilon)$-spanner with $O(n (\lg n)^{2})$ edges when $h = 0$ i.e., the polygonal domain is a simple polygon with no holes. 
Given a set $S$ of $n$ points on a polyhedral terrain $\mathcal{T}$, the geodesic Euclidean distance between any two points $p, q \in S$ is the distance along any shortest path between $p$ and $q$ on $\mathcal{T}$.
\cite{conf/soda/AbamBS17} showed that for a set of unweighted points on a polyhedral terrain, for any constant $\epsilon > 0$, there exists a $(2 + \epsilon)$-geodesic spanner with $O(n \lg n)$ edges.

For a network to be vertex fault-tolerant, i.e., when a subset of nodes is removed, the induced network on the remaining nodes requires to be connected.
Formally, a graph $G(S, E)$ is a {\it $k$-vertex fault-tolerant $t$-spanner}, denoted by $(k,t)$-VFTS, for a set $S$ of $n$ points in $\mathbb{R}^d$ if for any subset $S'$ of $S$ with size at most $k$, the graph $G \setminus S'$ is a $t$-spanner for the points in $S \setminus S'$.
Algorithms in Levcopoulos et~al., \cite{journals/algorithmica/LevcopoulosNS02}, Lukovszki \cite{conf/wads/Lukovszki99}, and Czumaj and Zhao \cite{journals/dcg/CzumajZ04} compute a $(k,t)$-VFTS for the set $S$ of points in $\mathbb{R}^d$.
These algorithms are also presented in \cite{books/compgeom/narsmid2007}.
\cite{journals/algorithmica/LevcopoulosNS02} devised an algorithm to compute a $(k,t)$-VFTS of size $O(\frac{n}{(t-1)^{(2d -1)(k+1)}})$ in $O(\frac{n \lg{n}}{(t-1)^{4d -1}} + \frac{n}{(t-1)^{(2d -1)(k+1)}})$ time 
and another algorithm to compute a $(k,t)$-VFTS with $O(k^{2}n)$ edges in $O(\frac{kn \lg n}{(t-1)^{d}})$ time. 
\cite{conf/wads/Lukovszki99} gives an algorithm to compute a $(k,t)$-VFTS of size $O(\frac{kn}{(t-1)^{d-1}})$ in $O(\frac{1}{(t-1)^{d}}(n \lg^{d-1} n \lg k + k n \lg \lg n))$ time. 
The algorithm in \cite{journals/dcg/CzumajZ04} computes a $(k,t)$-VFTS 
having $O(\frac{kn}{(t-1)^{d-1}})$ edges in $O(\frac{1}{(t-1)^{d-1}}(kn \lg^{d} n + nk^{2} \lg k))$ time with total weight of edges upper bounded by a $O(\frac{k^{2} \lg n}{(t-1)^{d}})$ multiplicative factor of the weight of MST of the given set of points. 

For a real number $t > 1$, a graph $G(S, E)$ is called a {\it $t$-spanner} for the weighted set $S$ of points if for any two points $p$ and $q$ in $S$ the distance between $p$ and $q$ in graph $G$ is at most $t$.$d_w(p, q)$ for a real number $t > 1$.
For some integer $k \geq 1$, a $t$-spanner $G$ for the set $S$ is a {\it $(k, t)$-vertex fault-tolerant additive weighted spanner}, denoted with $(k, t)$-VFTAWS, if for any set $S' \subset S$ with cardinality at most $k$, the graph $G \setminus S'$ is a $t$-spanner for the points in $S \setminus S'$.
\hfil\break

\noindent
{\bf Our results.}
The spanners computed in this paper are first of their kind as we combine fault-tolerance with the additive weighted set of points.
We devise the following algorithms for computing vertex fault-tolerant additive weighted geometric spanners (VFTAWS) for any $\epsilon > 0$ and $k \ge 1$:
\begin{itemize}

\item[*]
Given a set $S$ of $n$ weighted points in $\mathbb{R}^{d}$, our first algorithm presented herewith computes a $(k, 4+\epsilon)$-VFTAWS having $O(kn)$ edges.
We incorporate fault-tolerance to the recent results of \cite{conf/soda/AbamBS17} while retaining the same stretch factor and increasing the number of edges in the spanner by a multiplicative factor of $O(k)$. 

\item[*]
Given a set $S$ of $n$ weighted points in a simple polygon, we present an algorithm to compute a $(k, 4+\epsilon)$-VFTAWS that has $O(\frac{kn}{\epsilon^{2}}\lg{n})$ edges. 
Our algorithm combines the clustering based algorithms from \cite{conf/compgeom/AbamAHA15} and \cite{journals/algorithmica/AbamBFGS11}, and with the careful addition of more edges, we show that $k$ fault-tolerance is achieved.

\item[*]
Given a set $S$ of $n$ weighted points in a polygonal domain, we compute a $(k, 4+\epsilon)$-VFTAWS having $O(\frac{k\sqrt{h}n}{\epsilon^{2}}\lg{n})$ edges.
We extend the structures used in \cite{conf/compgeom/AbamAHA15} to achieve the vertex fault-tolerance. 

\item[*]
Given a set $S$ of $n$ weighted points on a terrain, we computes a $(k, 4+\epsilon)$-VFTAWS having $O(\frac{kn}{\epsilon^{2}}\lg{n})$ edges.
For achieving the fault-tolerance, our algorithm adds a minimal set of edges to the spanner constructed in \cite{conf/soda/AbamBS17}.
\end{itemize}

Unless specified otherwise, the points are assumed to be in Euclidean space $\mathbb{R}^d$.
The Euclidean distance between two points $p$ and $q$ is denoted by $|pq|$.
The distance between two points $p, q$ in the metric space $X$ is denoted by $d_X(p, q)$.
The length of the shortest path between $p$ and $q$ in a graph $G$ is denoted by $d_G(p, q)$.

Section~\ref{sect:rd} details the algorithm and its analysis to compute a $(k, 4+\epsilon)$-VFTAWS when the input weighted points are in $\mathbb{R}^d$.
For the input weighted points located in a simple polygon, Section~\ref{sect:simppoly} describes an algorithm to compute a $(k, 4+\epsilon)$-VFTAWS. 
Section~\ref{sect:polydom} details algorithm for $(k, 4+\epsilon)$-VFTAWS for a points located in the polygonal domain.
Section~\ref{sect:terrains} presents an algorithm to compute a $(k, 4+\epsilon)$-VFTAWS when the input points are located on a terrain.
Conclusions are in Section~\ref{sect:conclu}.

\section{Vertex fault-tolerant additive weighted spanner for points in $\mathbb{R}^d$}
\label{sect:rd}

In this section, we describe an algorithm to compute a $(k,t)$-VFTAWS for the set $S$ of $n$ non-negative weighted points in $\mathbb{R}^d$, where $t > 1$ and $k \geq 1$ are real numbers. 
For any two points $p, q \in S$, the additive weighted distance $d_{w}(p, q)$ is defined as the $w(p) + |pq| + w(q)$.
Following the algorithm mentioned in \cite{conf/soda/AbamBS17}, we first partition all the points belonging to $S$ into at least $k+1$ clusters.
For creating these clusters, the points in set $S$ are sorted in non-decreasing order with respect to their weights. 
Then the first $k+1$ points in this sorted list are chosen as the centers of $k+1$ distinct clusters. 
As the algorithm progress, more points are added to these clusters as well as more clusters (with cluster centers) may also be created.
In any iteration of the algorithm, for any point $p$ in the remaining sorted list, among the current set of cluster centers, we determine the cluster center $c_j$ nearest to $p$.
Let $C_j$ be the cluster to which $c_j$ is the center.
It adds $p$ to the cluster $C_j$ if $|p c_j| \leq \epsilon.w(p)$; otherwise, a new cluster $C_p$ with $p$ as its centre is initiated. 
Let $C=\{ c_1, \ldots, c_z \}$ be the final set of cluster centers obtained through this procedure.
For every $i \in [1, z]$, the cluster to which $c_i$ is the center is denoted by $C_i$.
Using the algorithm from \cite{conf/stoc/Solomon14}, we compute a $(k, (2+\epsilon))$-VFTS $\mathcal{B}$ for the set $C$ of cluster centers.
We note that the degree of each vertex of $\mathcal{B}$ is $O(k)$.
We denote the stretch of $\cal{B}$ by $t_{\cal{B}}$.
First, the graph ${\cal G}$ is initialized to $\mathcal{B}$; further, points in $S \setminus C$ are included in ${\cal G}$ as vertices.
Our algorithm to compute a $(k, 4+\epsilon)$-VFTAWS differs from \cite{conf/soda/AbamBS17} with respect to both the algorithm used in computing $\mathcal{B}$ and the set of edges added to ${\cal B}$. 
The latter part is described now.
For every $i \in [1, z]$, let $C_i'$ be the set comprising of $\min\{k+1, |C_i|\}$ least weighted points of cluster $C_i$. 
For each point $p \in S \setminus C$, if $p$ belongs to cluster $C_l$, then for each $v \in B_l \cup C_l'$, our algorithm introduces an edge between $p$ and $v$ with weight $|pv|$ to ${\cal G}$. 
Here, $B_l$ is the set comprising of all the neighbors of the center ($c_l$) of cluster $C_l$ in the graph $\cal{B}$.
We state our algorithm to compute a $(k,(2+\epsilon))$-vertex fault tolerant spanner(VFTS) for the set $S$.

\begin{algorithm}[H]
    
    \SetKwInOut{Input}{Input}
    \SetKwInOut{Output}{Output}
    
    \Input{The set $S$ of $n$ additive weighted points in $\mathbb{R}^d$ and an integer $k \geq 1$ and a real number $\epsilon$ lying in the range $[\frac{\lg k}{\lg \frac{n}{k+1} + 1},1]$ if $n \geq k(k+1)$ else $\epsilon$ is in the range $[\frac{k+1}{n}, 1]$
	}
    
    \Output{$(k,(4+\epsilon))$-VFTS $G$}
    
    \begin{algorithmic}[1]    
    
    \STATE create the list $P$ containing the points of $S$ sorted in non-decreasing order of their weights
    
    \STATE if $n \geq k(k+1)$, then create $\frac{n}{k^{\frac{1}{\epsilon}}}$ clusters; otherwise, create $\epsilon.n$ clusters
    
    \STATE assign the first point $q \in P$ to be the centre $c_1$ of the first cluster $C_1$ and set $z$ to $1$
    
    \STATE add $q$ to set $C$ of cluster centres
    
    \FOR{every point $p \in P \setminus \{ q \}$}
        
        \STATE find the cluster centre $c_j$ from the set $C$ nearest to $p$
        
        \IF{$|p c_j| \leq \epsilon. w(p)$}
            
            \STATE add $p$ to the cluster $C_j$
            
        \ELSE
            
            \STATE set $z$ as $z+1$
            
            \STATE create a new cluster $C_{z}$ and assign $p$ as the centre $c_{z}$ of $C_{z}$
            
            \STATE  add $p$ to the set $C$
            
        \ENDIF
        
    \ENDFOR
    
    \STATE construct a $(k,(1+\epsilon))$-VFTS $\mathcal{B}$ for the set $C$ using the algorithm given in \cite{conf/stoc/Solomon14}
    
    \FOR{every $p \in S \setminus C$}
        
        \STATE find the cluster $C_l$ to which $p$ belongs; let $c_l$ be the centre of $C_l$
        
	\STATE let $C_l'$ be the set of $min \{ k+1, card(C_l) \}$ least weighted points of cluster $C_l$; for every $p' \in C_l'$, introduce an edge between $p$ and $p'$
        
        \STATE let $B_l$ be the $k$ nearest neighbors of $c_l$ in $\mathcal{B}$; for every $p' \in B_l$, introduce an edge between $p$ and $p'$
    
    \ENDFOR
    
    \end{algorithmic}

    \caption{$k$-\large A\small DDITIVE\large FTS($S, k, \epsilon$)}
    \label{alg:addfts}
    
\end{algorithm}

In the following theorem, we prove that the graph ${\cal G}$ is indeed a $(k, 4+\epsilon)$-VFTAWS with $O(kn)$ edges.

\begin{theorem}
\label{thm:rd}
Let $S$ be a set of $n$ weighted points in $\mathbb{R}^d$ with non-negative weights associated to points with weight function $w$.
For any fixed constant $\epsilon > 0$, the graph ${\cal G}$ is a $(k, (4+\epsilon))$-VFTAWS with $O(kn)$ edges for the metric space $(S, d_w)$.
\end{theorem}
\begin{proof}
From \cite{conf/stoc/Solomon14}, the number of edges in $\mathcal{B}$ is $O(k \hspace{0.02in} |C|)$, which is essentially $O(kn)$.
Further, the degree of each node in $\mathcal{B}$ is $O(k)$.
From each point in $S \setminus C$, we are adding at most $O(k)$ edges.
Hence, the number of edges in ${\cal G}$ is $O(kn)$.

In proving that ${\cal G}$ is a $(k, 4+\epsilon)$-VFTAWS for the metric space $(S, d_w)$, we show that for any set $S' \subset S$ with $|S'| \leq k$ and for any two points $p, q \in S \setminus S'$ there exists a $(4+\epsilon)$-spanner path between $p$ and $q$ in ${\cal G} \setminus S'$. 
Following are the possible cases based on the role $p$ and $q$ play with respect to clusters formed and their centers:

\noindent
{\it Case 1}: Both $p$ and $q$ are cluster centres of two distinct clusters i.e., $p,q \in C$.

\noindent
Since $\mathcal{B}$ is a $(k, (2+\epsilon))$-VFTS for the set $C$,
{\setlength{\abovedisplayskip}{0pt}
\begin{flalign}
\hspace{6mm}d_{\mathcal{{\cal G}} \setminus S'}(p,q) &=d_{\mathcal{B} \setminus S'}(p,q)&&\nonumber\\
            &\leq t_{\mathcal{B}} \cdot d_w(p,q).&\nonumber
\end{flalign}}

\noindent
{\it Case 2}: Both $p$ and $q$ are in the same cluster $C_i$ and one of them, w.l.o.g., say $p$, is the centre of $C_i$.
\noindent
Since $p$ is the least weighted point in $C_i$, there exists an edge joining $p$ and $q$ in ${\cal G}$.
Hence,
{\setlength{\abovedisplayskip}{0pt}
\begin{flalign}
\hspace{6mm}d_{\mathcal{G} \setminus S'}(p,q) &= d_w(p,q).&\nonumber
\end{flalign}}

\noindent
{\it Case 3}: Both $p$ and $q$ are in the same cluster, say $C_i$; $p \ne c_i$, $q \ne c_i$; and, $c_i \notin S'$.
Then,

{\setlength{\abovedisplayskip}{0pt}
\begin{flalign}
\hspace{6mm}d_{\mathcal{G} \setminus S'}(p,q) &= d_w(p,c_i) + d_w(c_i, q)&&\nonumber\\
            &= w(p) + |p c_i| + w(c_i) + w(c_i) + |c_i q| + w(q)&&\nonumber\\
            &\leq w(p) + \epsilon \cdot w(p) + w(c_i) + w(c_i) + \epsilon \cdot w(q) + w(q)&&\nonumber\\
            &[\text{since a point} \ x \ \text{is added to cluster} \ C_l \ \text{only if} \ |x c_l| \leq \epsilon \cdot w(x)]&&\nonumber\\
            &\leq w(p) + \epsilon \cdot w(p) + w(p) + w(q) + \epsilon \cdot w(q) + w(q)&&\nonumber\\
            &[\text{since the points are sorted in the non-decreasing order of their weights and the first}&&\nonumber\\ 
            &\text{point added to any cluster is taken as its center}]&&\nonumber\\
            &= (2 + \epsilon) \cdot [w(p) + w(q)]&&\nonumber\\
            &< (2 + \epsilon) \cdot [w(p) + |pq| + w(q)]&&\nonumber\\
            &= (2 + \epsilon) \cdot d_w(p,q).&\nonumber
\end{flalign}}

\noindent
{\it Case 4}: Both $p$ and $q$ are in the same cluster, say $C_i$; $p \ne c_i$, $q \ne c_i$; and, $c_i \in S'$.\\
\noindent
In the case of $|C_i| \leq k$, there exists an edge between $p$ and $q$ in ${\cal G}$.
Hence, suppose that $|C_i| > k$.
Let $S''$ be the set of $k+1$ least weighted points from $C_i$.
If $p,q \in S''$ then there exists an edge between $p$ and $q$ in ${\cal G}$. 
If $p \in S''$ and $q \notin S''$ then as well there exists an edge between $p$ and $q$.
(Argument for the other case in which $q \in S''$ and $p \notin S''$ is analogous.)
Now consider the case in which both $p, q \notin S''$.  
Since $p$ and $q$ are connected to every point in $S''$ and $|S''|=k+1$, there exists an $r \in S''$ such that $r \notin S'$ and the edges $(p, r)$ and $(r, q)$ belong to ${\cal G} \setminus S'$.
Therefore,
{\setlength{\abovedisplayskip}{0pt}
\begin{flalign}
\hspace{6mm}d_{\mathcal{G} \setminus S'}(p,q) &= d_w(p,r) + d_w(r, q)&&\nonumber\\
            &= w(p) + |p r| + w(r) + w(r) + |r q| +w(q)&&\nonumber\\
            &\leq w(p) + |p c_i| + |c_i r| + w(r) + w(r) + |r c_i| + |c_i q| + w(q)&&\nonumber\\
            &\text{[by triangle inequality]}&&\nonumber\\
            &\leq w(p) + \epsilon \cdot w(p) + \epsilon \cdot w(r) + w(r) + w(r) + \epsilon \cdot w(r) + \epsilon \cdot w(q) + w(q)&&\nonumber\\
            &[\text{since a point} \ x \ \text{is added to cluster} \ C_l \ \text{only if} \ |x c_l| \leq \epsilon \cdot w(x)]&&\nonumber\\
            &\leq w(p) + \epsilon \cdot w(p) + \epsilon \cdot w(p) + w(p) + w(q) + \epsilon \cdot w(q) + \epsilon \cdot w(q) + w(q)&&\nonumber\\
            &\text{[since for any point the edges are added to the} \ k+1 \ \text{least weighted}&&\nonumber\\ &\text{points of the cluster to which it belongs]}&&\nonumber\\
            &= (2 + 2\epsilon) \cdot [w(p) + w(q)]&&\nonumber\\
            &< (2 + 2\epsilon) \cdot [w(p) + |pq| + w(q)]&\nonumber\\
            &= (2+2\epsilon) \cdot d_w(p,q).&\nonumber
\end{flalign}}

\noindent
{\it Case 5}: Points $p$ and $q$ belong to two distinct clusters, say $p \in C_i$ and $q \in C_j$. 
In addition, $p \neq c_i$ and $q \ne c_j$, and neither of the cluster centres belong to $S'$.

\noindent
Then,
{\setlength{\abovedisplayskip}{0pt}
\begin{flalign}
\hspace{6mm}d_{\mathcal{G} \setminus S'}(p,q) &= d_w(p,c_i) + d_{\mathcal{B}}(c_i, c_j) + d_w(c_j, q)&&\nonumber\\
            &= w(p) + |p c_i| + w(c_i) + d_{\mathcal{B}}(c_i, c_j) + w(c_j) + |c_j q| +w(q)&&\nonumber\\
            &\leq w(p) + \epsilon \cdot w(p) + w(c_i) + d_{\mathcal{B}}(c_i, c_j) + w(c_j) + \epsilon \cdot w(q) + w(q)&&\nonumber\\
            &[\text{since a point} \ x \ \text{is added to cluster} \ C_l \ \text{only if} \ |x c_l| \leq \epsilon \cdot w(x)]&&\nonumber\\
            &\leq (1 + \epsilon) \cdot [w(p) + w(q)] + w(c_i) + w(c_j) + t_{\mathcal{B}} \cdot d_w(c_i,c_j)&&\nonumber\\
            &[\text{since} \ \mathcal{B} \ \text{is a} \ (k,t_{\mathcal{B}})\text{-vertex fault-tolerant spanner for the set} \ C]&&\nonumber
\end{flalign}}

{\setlength{\abovedisplayskip}{0pt}
\begin{flalign}
            &\leq (1 + \epsilon) \cdot [w(p) + w(q)] + w(p) + w(q) + t_{\mathcal{B}} \cdot d_w(c_i,c_j)&&\nonumber\\
            &\text{[since the points are sorted in the non-decreasing order of their weights and the first}&&\nonumber\\ 
            &\text{point added to any cluster is taken as center of that cluster]}&&\nonumber\\
            &= (2 + \epsilon) \cdot [w(p) + w(q)] + t_{\mathcal{B}} \cdot [w(c_i) + |c_i c_j| + w(c_j)]&&\nonumber\\
            &\leq (2 + \epsilon) \cdot [w(p) + w(q)] + t_{\mathcal{B}} \cdot [w(p) + |c_i c_j| + w(q)]&&\nonumber\\
            &\text{[since the points are sorted in the non-decreasing order of their weights and the first}&&\nonumber\\ 
            &\text{point added to any cluster is taken as its center]}&&\nonumber\\
            &\leq (2 + \epsilon) \cdot [w(p) + w(q)] + t_{\mathcal{B}} \cdot [w(p) + w(q) + |c_i p| + |pq| + |q c_j|]&&\nonumber\\
            &\text{[by triangle inequality]}&&\nonumber\\
            &\leq (2 + \epsilon) \cdot [w(p) + w(q)] + t_{\mathcal{B}} \cdot [w(p) + w(q) + \epsilon \cdot w(p) + |pq| + \epsilon \cdot w(q)]&&\nonumber\\ &[\text{since a point} \ x \ \text{is added to cluster} \ C_l \ \text{only if} \ |x c_l| \leq \epsilon \cdot w(x)]&&\nonumber\\
            &= (2 + \epsilon) \cdot [w(p) + w(q)] + t_{\mathcal{B}} \cdot [(1 + \epsilon) \cdot [w(p) + w(q)] + |pq|]&\nonumber\\
           \hspace{22mm}&< (2 + \epsilon) \cdot [w(p) + w(q) + |pq|] + t_{\mathcal{B}} \cdot (1 + \epsilon) \cdot [w(p) + w(q) + |pq|]&&\nonumber\\
       &< t_{\mathcal{B}}(2 + \epsilon) \cdot [w(p) + w(q) + |pq|] \hspace{0.1in} \text{ when } t_B \ge (2+\epsilon) &&\nonumber\\
            &\text{[since each point has non-negative weight associated with it]}&&\nonumber\\
            &= t_{\mathcal{B}} \cdot (2 + \epsilon) \cdot d_w(p,q).&\nonumber
\end{flalign}}

\noindent
{\it Case 6}: Both the points $p$ and $q$ are in two distinct clusters, w.l.o.g., say $p \in C_i$ and $q \in C_j$, one of them, say $p$ is the centre of $C_i$ (i.e., $p = c_i$), and $c_j \notin S'$.
Then,
{\setlength{\abovedisplayskip}{0pt}
\begin{flalign}
\hspace{5mm}d_{\mathcal{G} \setminus S'}(p,q) &= d_{\mathcal{B}}(c_i,c_j) + d_w(c_j, q)&&\nonumber\\
            &\leq t_{\mathcal{B}} \cdot d_w(c_i, c_j) + d_w(c_j, q)&&\nonumber\\
            &[\text{since} \ \mathcal{B} \ \text{is a} \ (k,t_{\mathcal{B}})\text{-vertex fault-tolerant spanner for the set} \ C]&&\nonumber\\
            &= t_{\mathcal{B}} \cdot d_w(c_i, c_j) + w(c_j) + |c_j q| + w(q)&&\nonumber\\
            &= t_{\mathcal{B}} \cdot [w(c_i) + |c_i c_j| + w(c_j)] + w(c_j) + |c_j q| + w(q)&&\nonumber\\
            &\leq t_{\mathcal{B}} \cdot [w(p) + |c_i c_j| + w(q)] + w(q) + |c_j q| + w(q)&&\nonumber\\            
            &\text{[since the points are sorted in the non-decreasing order of their weights and the first}&&\nonumber\\ 
            &\text{point added to any cluster is taken as its center]}&&\nonumber\\
            &\leq t_{\mathcal{B}} \cdot [w(p) + |c_i c_j| + w(q)] + w(q) + \epsilon \cdot w(q) + w(q)&&\nonumber\\
            &[\text{since a point} \ x \ \text{is added to cluster} \ C_l \ \text{only if} \ |x c_l| \leq \epsilon \cdot w(x)]&&\nonumber\\
            &\leq t_{\mathcal{B}} \cdot [w(p) + |c_i p| + |pq| + |q c_j| + w(q)] + w(q) + \epsilon \cdot w(q) + w(q)&&\nonumber\\
            &\text{[by triangle inequality]}&&\nonumber\\ 
            &\leq t_{\mathcal{B}} \cdot [w(p) + \epsilon \cdot w(p) + |pq| + \epsilon \cdot w(q) + w(q)] + w(q) + \epsilon \cdot w(q) + w(q)&&\nonumber\\
            &[\text{since a point} \ x \ \text{is added to cluster} \ C_l \ \text{only if} \ |x c_l| \leq \epsilon.w(x)]&&\nonumber\\
            &= t_{\mathcal{B}} \cdot [(1 + \epsilon) \cdot [w(p) + w(q)] + |pq|] + (2 + \epsilon) \cdot w(q)&&\nonumber\\
            &\leq t_{\mathcal{B}} \cdot [(1 + \epsilon) \cdot [w(p) + w(q)] + |pq|] + (2 + \epsilon) \cdot [w(p) + w(q) + |pq|]&&\nonumber\\
            &\text{[since each point has non-negative weight associated with it]}&&\nonumber\\
        &\leq t_{\mathcal{B}} \cdot (2 + \epsilon) \cdot [w(p) + w(q) + |pq|] \hspace{0.1in} \text{ when } t_{\mathcal{B}} \ge (2+\epsilon) &&\nonumber\\
            &\leq t_{\mathcal{B}} \cdot (2 + \epsilon) \cdot d_w(p,q).&\nonumber
\end{flalign}}

\noindent
{\it Case 7}: Both the points $p$ and $q$ are in two distinct clusters, say $p \in C_i$ and $q \in C_j$; $p \neq c_i$, $q \neq c_j$; and, one of these centers, say $c_j$, belongs to $S'$ and the other center $c_i \notin S'$.
Since $q$ is connected to all the neighbors of $c_j$ in $\mathcal{B}$, for any neighbor $c_r$ of $c_j$ with $c_r \in C$ and $c_r \notin S'$, the edge $(q, c_r) $ belongs to ${\cal G} \setminus S'$.  
Therefore,
{\setlength{\abovedisplayskip}{0pt}
\begin{flalign}
\hspace{6mm}d_{\mathcal{G} \setminus S'}(p,q) &= d_w(p,c_i) + d_{\mathcal{B}}(c_i, c_r) + d_w(c_r, q)&&\nonumber\\
        &= w(p) + |p c_i| + w(c_i) + d_{\mathcal{B}}(c_i, c_r) + w(c_r) + |c_r q| + w(q)&\nonumber\\
        &\leq w(p) + \epsilon \cdot w(p) + w(c_i) + d_{\mathcal{B}}(c_i, c_r) + w(c_r) + |c_r q| + w(q)&&\nonumber\\
        &[\text{since a point} \ x \ \text{is added to cluster} \ C_l \ \text{only if} \ |x c_l| \leq \epsilon \cdot w(x)]&&\nonumber\\
        &\leq w(p) + \epsilon \cdot w(p) + w(c_i) + d_{\mathcal{B}}(c_i, c_r) + w(c_r) + |c_r c_j| + |c_j q| + w(q)&&\nonumber\\
        &\text{[by triangle inequality]}&\nonumber\\
        &\leq w(p) + \epsilon \cdot w(p) + w(c_i) + d_{\mathcal{B}}(c_i, c_r) + w(c_r) + |c_r c_j| + \epsilon \cdot w(q) + w(q)&&\nonumber\\
        &[\text{since a point} \ x \ \text{is added to cluster} \ C_l \ \text{only if} \ |x c_l| \leq \epsilon \cdot w(x)]&&\nonumber\\
        &\leq (1 + \epsilon) \cdot [w(p) + w(q)] + w(c_i) + d_{\mathcal{B}}(c_i, c_r) + w(c_r) + |c_j c_r|\label{eq1}.&   
\end{flalign}}

\noindent
Since $\mathcal{B}$ is a $(k,t_{\mathcal{B}})$-VFTS, there are at least $k+1$ vertex disjoint $t_{\mathcal{B}}$-spanner paths between $c_j$ and $c_i$ in $\mathcal{B}$. 
Suppose $c_r$ is the neighbor of $c_j$ in $\mathcal{B}$ such that one of these $k+1$ paths from $c_j$ to $c_i$ passes through $c_r$.
We have the following:
{\setlength{\abovedisplayskip}{0pt}
\begin{flalign}
&\hspace{5mm}d_{\mathcal{B}}(c_i , c_r) + d_w(c_r , c_j) < t_{\mathcal{B}} \cdot d_w(c_i , c_j)&&\nonumber\\
\hspace{5mm}&\Rightarrow d_{\mathcal{B}}(c_i , c_r) + w(c_r) + |c_r  c_j| + w(c_j) < t_{\mathcal{B}} \cdot d_w(c_i , c_j)&&\nonumber\\
\hspace{5mm}&\Rightarrow d_{\mathcal{B}}(c_i , c_r) + |c_r c_j| + w(c_r) < t_{\mathcal{B}} \cdot d_w(c_i , c_j) - w(c_j).\label{eq2}&
\end{flalign}}
Substituting (\ref{eq2}) in (\ref{eq1}),
{\setlength{\abovedisplayskip}{0pt}
\begin{flalign}
\hspace{5mm}d_{\mathcal{G} \setminus S'}(p,q) &< (1 + \epsilon) \cdot [w(p) + w(q)] + w(c_i) + t_{\mathcal{B}} \cdot d_w(c_i , c_j) - w(c_j)&&\nonumber\\
        &\leq (1 + \epsilon) \cdot [w(p) + w(q)] + w(c_i) + t_{\mathcal{B}} \cdot d_w(c_i , c_j)&&\nonumber\\
        &\text{[since weight associated with any point is non-negative]}&&\nonumber\\
        &=(1 + \epsilon) \cdot [w(p) + w(q)] + w(c_i) + t_{\mathcal{B}} \cdot [w(c_i) + |c_i c_j| + w(c_j)]&&\nonumber\\
        &\leq (1 + \epsilon) \cdot [w(p) + w(q)] + w(p) + t_{\mathcal{B}} \cdot [w(p) + |c_i c_j| + w(q)]&&\nonumber\\
        &\text{[since the points are sorted in the non-decreasing order of their weights and the first}&&\nonumber\\
        &\text{point added to any cluster is taken as its center]}&&\nonumber\\
        &\leq (1 + \epsilon) \cdot [w(p) + w(q)] + w(p) + t_{\mathcal{B}} \cdot [w(p) + |c_i p| + |pq| + |q c_j| + w(q)]&&\nonumber\\
        &\text{[by triangle inequality]}&&\nonumber\\
        &\leq (1 + \epsilon) \cdot [w(p) + w(q)] + w(p) + t_{\mathcal{B}} \cdot [w(p) + \epsilon \cdot w(p) + |pq| + \epsilon \cdot w(q) + w(q)]&&\nonumber\\
        &[\text{since a point} \ x \ \text{is added to cluster} \ C_l \ \text{only if} \ |x c_l| \leq \epsilon \cdot w(x)]&&\nonumber\\
        &\leq (2 + \epsilon) \cdot [w(p) + w(q)] + t_{\mathcal{B}} \cdot [(1 + \epsilon) \cdot [w(p) + w(q)] + |pq|]&&\nonumber\\
    &\leq t_{\mathcal{B}} \cdot [(2 + \epsilon) \cdot [w(p) + w(q)] + |pq|] \hspace{0.1in} \text{ when } t_{\mathcal{B}} \ge (2+\epsilon) &&\nonumber\\
        &\leq t_{\mathcal{B}} \cdot (2 + \epsilon) \cdot [w(p) + w(q)] + t_{\mathcal{B}}[w(p) + |pq| + w(q)]&&\nonumber\\
        &\text{[since weight associated with any point is non-negative]}&&\nonumber\\
        &\leq t_{\mathcal{B}} \cdot (2 + \epsilon) \cdot d_w(p,q).&\nonumber
\end{flalign}}

\noindent
{\it Case 8}: Points $p$ and $q$ are in two distinct clusters, say $p \in C_i$ and $q \in C_j$; $p \neq c_i$, $q \neq c_j$; and both $c_i, c_j \in S'$.

\noindent
Since $p$ (resp. $q$) is connected to $k$ nearest neighbor of $c_i$ (resp. $c_j$), there exists $c_r, c_l \in C$ such that $c_r,c_l \notin S'$ and the edges $(p, c_r)$ and $(c_l, q)$ belong to ${\cal G} \setminus S'$.
Then,

{\setlength{\abovedisplayskip}{0pt}
\begin{flalign}
\hspace{5mm}d_{\mathcal{G} \setminus S'}(p,q) &= d_w(p,c_r) + d_{\mathcal{B}}(c_r, c_l) + d_w(c_l, q)&&\nonumber\\
        &= w(p) + |p c_r| + w(c_r) + d_{\mathcal{B}}(c_r, c_l) + w(c_l) + |c_l q| + w(q)&\nonumber\\
\hspace{25mm}&\leq w(p) + |p c_i| + |c_i c_r| + w(c_r) + d_{\mathcal{B}}(c_r, c_l)&\nonumber\\
    &\hspace{0.2in}+ w(c_l) + |c_j c_l| + |c_j q| + w(q)\label{eq3}&&\\
        &\text{[by triangle inequality]}.&\nonumber
\end{flalign}}

\noindent
Since $\mathcal{B}$ is a $(k,t_{\mathcal{B}})$-VFTS, there are at least $k+1$ vertex disjoint $t_{\mathcal{B}}$-spanner paths between $c_j$ and $c_i$ in $\mathcal{B}$.  
Suppose $c_r$ (resp. $c_l$) is the neighbor of $c_i$ (resp. $c_j$) in $\mathcal{B}$ such that one of these $k+1$ paths from $c_j$ to $c_i$ passes through $c_r$ (rsp. $c_l$).
We have the following:

{\setlength{\abovedisplayskip}{0pt}
\begin{flalign}
&\hspace{5mm}d_w(c_i,c_r) + d_{\mathcal{B}}(c_r,c_l) + d_w(c_l,c_j) < t_{\mathcal{B}} \cdot d_w(c_i,c_j)&&\nonumber\\
\hspace{5mm}&\Rightarrow w(c_r) + |c_i c_r| + w(c_i) + d_{\mathcal{B}}(c_r , c_l) + w(c_j) + |c_j c_l| + w(c_l) < t_{\mathcal{B}} \cdot d_w(c_i , c_j)&&\nonumber\\
\hspace{5mm}&\Rightarrow w(c_r) + |c_i c_r| + d_{\mathcal{B}}(c_r , c_l) + w(c_l) + |c_j c_l|  < t_{\mathcal{B}} \cdot d_w(c_i , c_j) - w(c_i) - w(c_j).\label{eq4}&
\end{flalign}}
Substituting (\ref{eq4}) in (\ref{eq3}), we get
{\setlength{\abovedisplayskip}{0pt}
\begin{flalign}
\hspace{5mm}d_{\mathcal{G} \setminus S'}(p,q) &< w(p) + |p c_i| + t_{\mathcal{B}} \cdot d_w(c_i , c_j) - w(c_i) - w(c_j) + |c_j q| + w(q)&&\nonumber\\
        &\leq w(p) + |p c_i| + t_{\mathcal{B}} \cdot d_w(c_i , c_j) + |c_j q| + w(q)&&\nonumber\\
        &\text{[since the weight associated with each point is non-negative]}&&\nonumber\\
        &= w(p) + |p c_i| + t_{\mathcal{B}} \cdot [w(c_i) + |c_i c_j| + w(c_j)] + |c_j q| + w(q)&&\nonumber\\
        &\leq w(p) + |p c_i| + t_{\mathcal{B}} \cdot [w(c_i) + |c_i p| + |pq| + |q c_j| + w(c_j)] + |c_j q| + w(q)&&\nonumber\\
        &\text{[by triangle inequality]}&&\nonumber\\
        &\leq w(p) + \epsilon \cdot w(p) + t_{\mathcal{B}} \cdot [w(c_i) + \epsilon \cdot w(p) + |pq| + \epsilon \cdot w(q) + w(c_j)] + \epsilon \cdot w(q) + w(q)&&\nonumber\\
        &[\text{since a point} \ x \ \text{is added to cluster} \ C_l \ \text{only if} \ |x c_l| \leq \epsilon \cdot w(x)]&&\nonumber\\
        &\leq w(p) + \epsilon \cdot w(p) + t_{\mathcal{B}} \cdot [w(p) + \epsilon \cdot w(p) + |pq| + \epsilon \cdot w(q) + w(q)] + \epsilon \cdot w(q) + w(q)&&\nonumber\\
        &[\text{since the points are sorted in the non-decreasing order of their weights and the first}&&\nonumber\\
        &\text{point added to any cluster is taken as its center}]&&\nonumber\\
        &\leq (1 + \epsilon) \cdot [w(p) + w(q)] + t_{\mathcal{B}} \cdot [(1 + \epsilon) \cdot [w(p) + w(q)] + |pq|]&&\nonumber\\
    &\leq t_{\mathcal{B}} \cdot (2 + \epsilon) \cdot [w(p) + w(q) + |pq|] \hspace{0.1in} \text{ when } t_{\mathcal{B}} \ge (1+\epsilon) &&\nonumber\\
        &= t_{\mathcal{B}} \cdot (2 + \epsilon) \cdot d_w(p,q).&\nonumber
\end{flalign}}

\noindent
Considering the analysis in all these cases proves that ${\cal G}$ is a $k$-VFTAWS with stretch $t$ upper bounded by $t_{\mathcal{B}} \cdot (2 + \epsilon)$.
We had chosen $t_{\mathcal{B}}$ to be equal to $(2 + \epsilon)$, so that it satisfies all the above cases.
Since $t_{\mathcal{B}}$ is $(2 + \epsilon)$, $t = (2 + \epsilon)^2 \le (4 + 5 \epsilon)$.
Hence, ${\cal G}$ is a $(k, 4 + 5\epsilon)$-VFTAWS for the metric space $(S, d_w)$.
\end{proof}

\ignore {
\begin{lemma}
\label{lem3}
The time taken to compute ${\cal G}$ is $O(n(n+k))$.
\end{lemma}

\begin{proof}
We need to account for the time required at each step of our algorithm. Sorting the points of $S$ in non-decreasing order of their weights takes $O(n \lg n)$ time. 
Creating the clusters for the set $S$ and storing the points belonging to a cluster with each cluster can be done in $O(n z)$ time. 
Here, $z$ is the number of clusters created. 
As stated in \cite{conf/stoc/Solomon14}, computing $\mathcal{B}$ takes $O(z(\lg z + k))$. 
Connecting each point $p \notin C$ to $min\{k+1, |C_l|\}$ least weighted points of the cluster $C_l$ to which it belongs takes $O(k (n - z))$ time. 
Since for each such $p$, we need to scan once, the first $k+1$ points of the cluster to which it belongs. 
Then connecting each point $p \notin C$ to $k$ nearest neighbors of $c_l$ in $\mathcal{B}$ can be done by doing a breadth-first traversal (BFS) of the graph $\mathcal{B}$. 
Time needed in this step is $O(kz)$, since $\mathcal{B}$ has $O(kz)$ edges. Since $z \leq n$, so a total of $O(kn)$ time is needed at this step.

The algorithm takes a total time of $O(n(n+k))$.
\end{proof}
}

\section{Vertex fault-tolerant additive weighted spanner for points in simple polygon} 
\label{sect:simppoly}

Given a set $S$ of $n$ points in a simple polygon $P$, for any two points $p, q \in S$, the shortest path between $p$ and $q$ in $P$ is denoted by $\pi(p, q)$, and the length of that path is indicated by $d_\pi(p, q)$.
For a $t \ge 1$, a geodesic $t$-spanner of $S$ is a graph ${\cal G}(S, E')$ such that $d_\pi(p, q) \le d_{\cal G}(p, q) \le t \cdot d_\pi(p, q)$ for every two points $p, q \in S$.
We detail an algorithm to compute a geodesic vertex fault-tolerant additive weighted spanner for the set $S$ of $n$ weighted points located in a simple polygon $P$. 

The following definition for the distance function $d_{\pi, w}$ for the set $S$ of points is considered in \cite{journals/algorithmica/AbamBFGS11}: For any $p,q \in S$, $d_{\pi,w}(p, q)$ equals to $0$ if $p = q$; otherwise, it is equal to $w(p) + d_\pi(p,q) + w(q)$.
Further, $(S, d_{\pi,w})$ was shown as a metric space in \cite{journals/algorithmica/AbamBFGS11}.
\ignore {
\[
d_{\pi,w}(p, q) =
\begin{cases}
    0 & \text{if } p = q. \\
    w(p) + d_\pi(p,q) + w(q) & \text{if } p \ne q.
\end{cases}
\]
}

We devise a divide-and-conquer based algorithm to compute a $(k, 4+\epsilon)$-VFTAWS for the metric space $(S, d_{\pi,w})$.
Following \cite{conf/compgeom/AbamAHA15}, we define few terms. 
Let $S'$ be a set of points contained in a simple polygon $P'$.
A vertical line segment that splits $P'$ into two simple sub-polygons of $P'$ such that each sub-polygon contains at most two-thirds of the points in $S'$ is termed a {\it splitting segment} with respect to $S'$ and $P'$.
(In the following description, $S'$ and $P'$ are not mentioned with the splitting segment whenever they are clear from the context.)
The {\it geodesic projection} $p_l$ of a point $p$ onto a splitting segment $l$ is a point on $l$ that has the minimum geodesic Euclidean distance from $p$ among all the points of $l$.
By extending \cite{conf/compgeom/AbamAHA15}, we give an algorithm to compute a $(k, 4+\epsilon)$-VFTAWS ${\cal G}$ for the metric space $(S, d_{\pi, w})$.

Our algorithm partitions $P$ containing points in $S$ into two simple sub-polygons $P'$ and $P''$ with a splitting segment $l$. 
For every point $p \in S$, we compute its geodesic projection $p_l$ onto $l$ and assign $w(p) + d_\pi(p, p_l)$ as the weight of $p_l$.
Let $S_l$ be the set comprising of all the geodesic projections of $S$ onto $l$.
Also, let $d_{l, w}$ be the additive weighted metric associated with points in $S_l$.
We use the algorithm from Section~\ref{sect:rd} to compute a $(k, 4+\epsilon)$-VFTAWS ${\cal G}_l$ for the metric space $(S_l, d_{l, w})$.
For every edge $(r, s)$ in ${\cal G}_l$, we add an edge between $p$ and $q$ to ${\cal G}$ with weight $d_\pi(p, q)$, wherein $r$ (resp. $s$) is the geodesic projection of $p$ (resp. $q$) onto $l$.
Let $S'$ (resp. $S''$) be the set of points contained in the sub-polygon $P'$ (resp. $P''$) of $P$.
We recursively process $P'$ (resp. $P''$) with points in $S'$ (resp. $S''$) unless $|S'|$ (resp. $|S''|$) is less than or equal to one.

\begin{algorithm}[H]     
    
    \SetKwInOut{Input}{Input}
    \SetKwInOut{Output}{Output}

    \caption{$k$-\large A\small DDITIVE\large P\small OLYGON\large C\small LUST\large FTS($S, P, k, \epsilon$)}
    \label{alg:addpolyclustfts}
    
    \Input{The simple polygon $P$, the set $S$ on $n$ additive weighted points, integer $k \geq 1$ and a real number $0 < \epsilon \leq 1$}
    
    \Output{$(k,(12+\epsilon))$-VFTS $\mathcal{G}$}

    \begin{algorithmic}[1]
    
    \WHILE{$card(P \cap S) \geq 1$}
        
        \STATE compute a splitting segment $l$ for $P$ using the algorithm given in \cite{conf/jcdcg/BoseCKKM98}
    
        \STATE initialize the set $S_l$ of the projections to $\phi$
    
        \FOR{every $p \in S$}
    
            \STATE find the projection $p_l$ of $p$ on $l$
            
            \STATE assign a weight $w(p) + d_{\pi}(p,p_l)$ to $p_l$
        
            \STATE $S_l := S_l \cup \{p_l\}$
        
        \ENDFOR
    
        \STATE compute a $(k,(4+\epsilon))$-vertex fault-tolerant spanner $\mathcal{G}_l$ for the set $S_l$ using Algorithm \ref{alg:addfts}
        
        \STATE for every edge $(r,s)$ in $\mathcal{G}_l$, add the edge $(p,q)$ to $\mathcal{G}$, where \\
	    $r$(resp. $s$) is the projection of $p$(resp. $q$) on $l$
    
        \STATE $k$-\large A\small DDITIVE\large P\small OLYGON\large C\small LUST\large FTS($S \cap P_l', P_l', k, \epsilon$)
        
        \small  //$P_l'$ denotes the sub-polygon to the left of $l$
    
        \STATE $k$-\large A\small DDITIVE\large P\small OLYGON\large FTS($S\cap P_l'', P_l'', k, \epsilon$)
        
        \small  //$P_l''$ denotes the sub-polygon to the right of $l$
        
    \ENDWHILE
    \end{algorithmic}

\end{algorithm}

We prove that the graph ${\cal G}$ is a $(k,(12+15\epsilon))$-VFTAWS for the metric space $(S, d_{\pi, w})$.
(Later, with further refinements to this graph, we improve the stretch factor to $(4+14\epsilon)$.)
We show that by removing any subset $S'$ with $|S'| \leq k$ from ${\cal G}$, for any two points $p$ and $q$ in $S \setminus S'$, there exists a path between $p$ and $q$ in ${\cal G} \setminus S'$ such that the $d_{\cal G}(p, q)$ is at most $(12+15\epsilon)d_{\pi, w}(p, q)$.
First, we note that there exists a splitting segment $l$ at some iteration of the algorithm so that $p$ and $q$ are on different sides of $l$. 
Let $r$ be a point belonging to $l \cap \pi(p, q)$.
Let $S'_l$ be the set comprising of geodesic projections of points in $S'$ on $l$.
Since ${\cal G}_l$ is a $(k,(4+5\epsilon))$-VFTAWS for the metric space $(S_l, d_{l, w})$, there exists a path $Q$ between $p_l$ and $q_l$ in ${\cal G}_l \setminus S'_l$ whose length is upper bounded by $(4+5\epsilon) \cdot d_{l,w}(p_l, q_l)$. 
Let $Q'$ be a path between $p$ and $q$ in ${\cal G} \setminus S'$ which is obtained by replacing each vertex $v_l$ of $Q$ by $v$ in $S$ such that the point $v_l$ is the geodesic projection of $v$ on $l$.
In the following, we show that the length of $Q'$, which is $d_{{\cal G} \setminus S'}(p, q)$, is upper bounded by $(12+15\epsilon) \cdot d_{\pi,w}(p, q)$.

\noindent
For every $x,y \in S$,
{\setlength{\abovedisplayskip}{0pt}
\begin{flalign}
\hspace{6mm}d_{\pi,w}(x,y) &= w(x) + d_{\pi}(x,y) + w(y)&&\nonumber\\
        &\leq w(x) + d_{\pi}(x,x_l) + d_{\pi}(x_l,y_l) + d_{\pi}(y_l,y) + w(y)&&\nonumber\\
        &\text{[by triangle inequality]}&&\nonumber\\
        &= w(x_l) + d_{\pi}(x_l,y_l) + w(y_l)&&\nonumber\\
        &\text{[since the weight associated with projection} \ z_l \ \text{of every point} \ z \ \text{is} \ w(z) + d_{\pi}(z,z_l)]&&\nonumber\\
        &=d_{l,w}(x_l,y_l).\label{eq18}&
\end{flalign}}

\noindent
This implies,
{\setlength{\abovedisplayskip}{0pt}
\begin{flalign}
\hspace{6mm}d_{{\cal G} \setminus S'}(p,q) &\le \sum_{x_l,y_l \in Q} d_{\pi,w}(x,y)&&\nonumber\\
        &\leq \sum_{x_l,y_l \in Q} d_{l,w}(x_l,y_l)&&\nonumber\\
        &\text{[\text{from} \ (\ref{eq18})]}&&\nonumber\\
        \hspace{22mm}&\leq (4+5\epsilon) \cdot d_{l,w}(p_l,q_l)\label{eq19}&&\\
        &\text{[since} \ {\cal G}_l \ \text{is a} \ (k,(4+5\epsilon))\text{-vertex fault-tolerant geodesic spanner]}&&\nonumber\\
        &= (4+5\epsilon) \cdot [w(p_l) + d_l(p_l,q_l) + w(q_l)]&&\nonumber\\
        &= (4+5\epsilon) \cdot [w(p_l) + d_{\pi}(p_l,q_l) + w(q_l)]&&\nonumber\\
        &\text{[since} \ P \ \text{contains $l$, shortest path between } p_l \text{ and } q_l&&\nonumber\\
        &\text{is same as the geodesic shortest path between } p_l \text{ and } q_l]&&\nonumber\\
        &= (4+5\epsilon) \cdot [w(p) + d_{\pi}(p,p_l) + d_{\pi}(p_l,q_l) + d_{\pi}(q_l,q) + w(q)]\label{eq20}&&\\
        &\text{[since the weight associated with projection} \ z_l \ \text{of every point} \ z \ \text{is} \ w(z) + d_{\pi}(z,z_l) \text{]}.&\nonumber
\end{flalign}}

\noindent
Since $r$ is a point belonging to both $l$ as well as to $\pi(p, q)$,
{\setlength{\abovedisplayskip}{0pt}
\begin{flalign}
\hspace{6mm}&d_{\pi}(p,p_l) \leq d_{\pi}(p,r) \ \text{and} \ d_{\pi}(q,q_l) \leq d_{\pi}(q,r).\label{eq21}&
\end{flalign}}

\noindent
Substituting (\ref{eq21}) into (\ref{eq20}),
{\setlength{\abovedisplayskip}{0pt}
\begin{flalign}
\hspace{6mm}d_{\mathcal{G} \setminus S'}(p,q) &\leq (4+5\epsilon) \cdot [w(p) + d_{\pi}(p,r) + d_{\pi}(p_l,q_l) + d_{\pi}(r,q) + w(q)]&&\nonumber\\
        &\leq (4+5\epsilon) \cdot [w(p) + d_{\pi}(p,r) + w(r) + d_{\pi}(p_l,q_l) + w(r) + d_{\pi}(r,q) + w(q)]&&\nonumber\\
        &\text{[since the weight associated with every point is non-negative]}&&\nonumber\\
        &= (4+5\epsilon) \cdot [d_{\pi,w}(p,r) + d_{\pi}(p_l,q_l) + d_{\pi,w}(r,q)]&&\nonumber\\
        &= (4+5\epsilon) \cdot [d_{\pi,w}(p,q) + d_{\pi}(p_l,q_l)]\label{eq22}&&\\
        &\text{[since} \ \pi(p,q) \ \text{intersects} \ l \text{ at} \ r\text{, by optimal substructure property of shortest}&&\nonumber\\ 
        &\text{paths,} \ \pi(p,q) = \pi(p,r) + \pi(r,q)\text{]}.&\nonumber
\end{flalign}}

\noindent
Consider
{\setlength{\abovedisplayskip}{0pt}
\begin{flalign}
\hspace{6mm}d_{\pi}(p_l,q_l) &\leq d_{\pi}(p_l,p) + d_{\pi}(p,q) + d_{\pi}(q,q_l)&&\nonumber\\
        &\text{[since} \ \pi \ \text{follows triangle inequality]}&&\nonumber\\
        &\leq d_{\pi}(r,p) + d_{\pi}(p,q) + d_{\pi}(q,r)&&\nonumber\\
        &\text{[using} \ (\ref{eq21})\text{]}&&\nonumber\\
        &\leq w(r) + d_{\pi}(r,p) + w(p) + w(p) + d_{\pi}(p,q) + w(q) + w(q) + d_{\pi}(q,r) + w(r)&&\nonumber\\
        &\text{[since weight associated with every point is non-negative]}&&\nonumber\\
        &= d_{\pi,w}(p,r) + d_{\pi,w}(p,q) + d_{\pi,w}(r,q)&&\nonumber\\
        &= d_{\pi,w}(p,q) + d_{\pi,w}(p,q)&&\nonumber\\
        &\text{[since} \ \pi(p,q) \ \text{intersects} \ l \text{ at } \ r\text{, by optimal substructure property of shortest}&&\nonumber\\ 
        &\text{paths,} \ \pi(p,q) = \pi(p,r) + \pi(r,q)\text{]}&&\nonumber\\
        &= 2d_{\pi,w}(p,q).\label{eq23}&
\end{flalign}}
Substituting (\ref{eq23}) into (\ref{eq22}),
{\setlength{\abovedisplayskip}{0pt}
\begin{flalign}
\hspace{6mm}&d_{G \setminus S'}(p,q) \leq 3(4+5\epsilon) \cdot d_{\pi,w}(p,q).&\nonumber
\end{flalign}}


Hence, the graph $\mathcal{G}$ computed as described above is a $(k, 12+\epsilon)$-VFTAWS for the metric space $(S, d_{\pi, w})$.
We further improve the stretch factor of $\mathcal{G}$ by applying the refinement given in \cite{conf/soda/AbamBS17} to the above-described algorithm.
In doing this, for each point $p \in S$, we compute the geodesic projection $p_{\gamma}$ of $p$ on the splitting line $\gamma$ and we construct a set $S(p,\gamma)$ as defined herewith. 
Let $\gamma(p) \subseteq \gamma$ be $\{ x \in \gamma : d_{\gamma,w}(p_{\gamma},x) \leq (1 + 2\epsilon) \cdot d_{\pi}(p,p_{\gamma}) \}$. 
Here, for any $p,q \in S$, $d_{\gamma ,w}(p, q)$ is equal to $0$ if $p = q$; otherwise, equals to $w(p) + d_\gamma(p,q) + w(q)$.
We divide $\gamma(p)$ into $c$ pieces with $c \in O(1/\epsilon^{2})$: each piece is denoted by $\gamma_{j}(p)$ for $1 \leq j \leq c$, and the piece length is at most $\epsilon \cdot d_{\pi}(p,p_{\gamma})$. 
For every piece $j$, we compute the point $p_\gamma^{(j)}$ nearest to $p$ in $\gamma_j(p)$.
The set $S(p,\gamma)$ is defined as $\{ p_\gamma^{(j)} : p_\gamma^{(j)} \in \gamma_{j}(p)$  and $1 \leq j \leq c\}$. 
For every $r \in S(p,\gamma)$, the non-negative weight $w(r)$ of $r$ is set to $w(p) + d_{\pi}(p,r)$.
Let $S_\gamma$ be $\cup_{p \in S}S(p,\gamma)$.

We replace the set $S_l$ in computing ${\cal G}$ with the set $S_\gamma$ and compute a $(k,(4+5\epsilon))$-VFTAWS $\mathcal{G}_l$ using the algorithm from Section~\ref{sect:rd} for the set $S_\gamma$ instead. 
Further, for every edge $(r,s)$ in $\mathcal{G}_l$, we add the edge $(p,q)$ to $\mathcal{G}$ with weight $d_\pi(p, q)$ whenever $r \in S(p,l)$ and $s \in S(q,l)$.
The rest of the algorithm remains the same.

In the following, we restate a lemma from \cite{conf/soda/AbamBS17}, which is useful for our analysis.

\begin{lemma}
\label{lemfromabam}
Let $P$ be a simple polygon.
Consider two points $x,y \in P$.
Let $r$ be the point at which shortest path $\pi(x, y)$ between $x$ and $y$ intersects a splitting segment $\gamma$.
If $r \notin \gamma(x)$, point $x_\gamma'$ (resp. $y_\gamma'$) is set as $x_\gamma$ (resp. $y_\gamma$). 
Otherwise $x_\gamma'$ (resp. $y_\gamma'$) is set as the point from $S(x,\gamma)$ (resp. $S(y,\gamma)$) which is nearest to $x$ (resp. $y$). 
Then $d_{\pi}(x,x_\gamma') + d_\gamma(x_\gamma',r)$ (resp. $d_\gamma(r,y_\gamma') + d_{\pi}(y_\gamma',y)$) is less than or equal to $(1 + \epsilon) \cdot d_{\pi}(x,r)$ (resp. $(1 + \epsilon) \cdot d_{\pi}(r,y)$).
\end{lemma}

\begin{theorem}
\label{thm:simppoly}
Let $S$ be a set of $n$ weighted points in simple polygon $P$ with non-negative weights associated to points with weight function $w$.
For any fixed constant $\epsilon > 0$, there exists a $(k, (4+\epsilon))$-vertex fault-tolerant additive weighted geodesic spanner with $O(\frac{kn}{\epsilon^2} \lg{n})$ edges for the metric space $(S, d_{\pi, w})$.
\end{theorem}

\begin{proof}
In constructing a $(k,(4+\epsilon))$-VFTAWS ${\mathcal{G}}_l$ for the set $S_\gamma$ of $\frac{n}{\epsilon^{2}}$ points, we add $O(\frac{k n}{\epsilon^{2}})$ edges to $\mathcal{G}$ in one iteration.
Let $S(n)$ be the size of $\mathcal{G}$ when there are $n$ points. 
Then $S(n) = S(n_1) + S(n_2) + \frac{k n}{\epsilon^{2}}$ where $n_1, n_2$ are the number of points in each of the partitions formed by the splitting segment.
Since $n_1, n_2 \geq n/3$, $S(n)$ is $O(\frac{k n}{\epsilon^{2}} \lg n)$.

For proving that ${\cal G}$ is a $(k, (4+\epsilon))$-VFTAWS for the metric space $(S, d_{\pi, w})$, we show that for any set $S' \subset S$ with $|S'| \leq k$ and for any two points $p, q \in S \setminus S'$ there exists a $(4+\epsilon)$-spanner path between $p$ and $q$ in ${\cal G} \setminus S'$. 
First, we note that there exists a splitting segment $l$ at some iteration of the algorithm so that $p$ and $q$ are on different sides of $l$.
Let $r$ be a point belonging to $l \cap \pi(p, q)$.
Let $S'_l$ be the set comprising of geodesic projections of points in $S'$ on $l$.
Since ${\cal G}_l$ is a $(k,(4+5\epsilon))$-VFTAWS for the metric space $(S_l, d_{l, w})$, there exists a path $Q$ between $p_l$ and $q_l$ in ${\cal G}_l \setminus S'_l$ whose length is upper bounded by $(4+5\epsilon) \cdot d_{l,w}(p_l, q_l)$. 
Let $Q'$ be a path between $p$ and $q$ in ${\cal G} \setminus S'$ which is obtained by replacing each vertex $v_l$ of $Q$ by $v$ in $S$ such that the point $v_l$ is the geodesic projection of $v$ on $l$.
In the following, we show that the length of $Q'$, which is $d_{{\cal G} \setminus S'}(p, q)$, is upper bounded by $(4+14\epsilon) \cdot d_{\pi,w}(p, q)$.

\ignore {
For every $x,y \in S$,
{\setlength{\abovedisplayskip}{0pt}
\begin{flalign}
\hspace{6mm}d_{\pi,w}(x,y) &= w(x) + d_{\pi}(x,y) + w(y)&&\nonumber\\
        &\leq w(x) + d_{\pi}(x,x_l) + d_{\pi}(x_l,y_l) + d_{\pi}(y_l,y) + w(y)&&\nonumber\\
        &\text{[by triangle inequality]}&&\nonumber\\
        &= w(x_l) + d_{\pi}(x_l,y_l) + w(y_l)&&\nonumber\\
        &\text{[since the weight associated with projection} \ z_l \ \text{of every point} \ z \ \text{is } \ w(z) + d_{\pi}(z,z_l)]&&\nonumber\\
        &=d_{l,w}(x_l,y_l)\label{eq18}&
\end{flalign}}

\noindent
This implies, 
{\setlength{\abovedisplayskip}{0pt}
\begin{flalign}
\hspace{6mm}d_{{\cal G} \setminus S'}(p,q) &\le \sum_{x_l,y_l \in Q} d_{\pi,w}(x,y)&&\nonumber\\
        &\leq \sum_{x_l,y_l \in Q} d_{l,w}(x_l,y_l)&&\nonumber\\
        &\text{[\text{from} \ (\ref{eq18})]}&&\nonumber\\
        \hspace{22mm}&\leq (2+6\epsilon) \cdot d_{l,w}(p_l,q_l)\label{eq19}&&\\
        &\text{[since} \ {\cal G}_l \ \text{is a} \ (k,(2 + 6\epsilon))\text{-vertex fault-tolerant geodesic spanner]}&&\nonumber
\end{flalign}}
}

\noindent
Following Lemma~\ref{lemfromabam}, if $r \notin l(p)$, point $p_l'$ (resp. $q_l'$) is set as $p_l$ (resp. $q_l$). 
Otherwise $p_l'$ (resp. $q_l'$) is set as the point from $S(p,l)$ (resp. $S(q,l)$) which is nearest to $p$ (resp. $q$).

{\setlength{\abovedisplayskip}{0pt}
\begin{flalign}
\hspace{6mm}d_{l,w}(p_l',q_l') &= w(p_l') + d_l(p_l',q_l') + w(q_l')&&\nonumber\\
        &\leq w(p_l') + d_l(p_l',r) + d_l(r,q_l') + w(q_l')&&\nonumber\\
        &\text{[by triangle inequality]}&\nonumber\\
        &\leq w(p_l') + d_l(p_l',r) + w(r) + w(r) +  d_l(r,q_l') + w(q_l')&&\nonumber\\
        &\text{[since the weight associated with each point is non-negative]}&&\nonumber\\
        &= w(p) + d_{\pi}(p,p_l') + d_l(p_l',r) + w(r)&&\nonumber \\
    &\hspace{0.3in} + w(r) + d_l(r,q_l') + d_{\pi}(q_l',q) + w(q)\label{eq24}&&\\
        &\text{[due to weight assigned to geodesic projections]}.&\nonumber
\end{flalign}}

\noindent
Applying Lemma~\ref{lemfromabam} with $p_l'$ and $q_l'$,
{\setlength{\abovedisplayskip}{0pt}
\begin{flalign}
\hspace{6mm}&d_{\pi}(p,p_l') + d_l(p_l',r) \leq (1 + \epsilon) \cdot d_{\pi}(p,r)\text{, and} \label{eq25'}\\
\hspace{6mm}&d_l(r,q_l') + d_{\pi}(q_l',q) \leq (1 + \epsilon) \cdot d_{\pi}(r,y).\label{eq25}&
\end{flalign}}

\noindent
Substituting (\ref{eq25'}) and (\ref{eq25}) in (\ref{eq24}),
{\setlength{\abovedisplayskip}{0pt}
\begin{flalign}
\hspace{6mm}d_{l,w}(p_l',q_l') &\leq w(p) + (1 + \epsilon) \cdot d_\pi(p,r) + w(r) + w(r) + (1 + \epsilon) \cdot d_\pi(r,q) + w(q)&&\nonumber\\
        &\leq (1 + \epsilon) \cdot [d_{\pi,w}(p,r) + d_{\pi,w}(r,q)]&&\nonumber\\
        &= (1 + \epsilon) \cdot d_{\pi,w}(p,q)\label{eq26}&&\\
        &[\text{since } r \in l \cap \pi(p,q), \text{ by the optimal substructure property of shortest}&&\nonumber\\ 
        &\text{paths,} \ \pi(p,q) = \pi(p,r) + \pi(r,q)].&\nonumber
\end{flalign}}

\noindent
Replacing $p_l$ (resp. $q_l$) by ${p_l}'$ (resp. ${q_l}'$) in inequality (\ref{eq19}),
{\setlength{\abovedisplayskip}{0pt}
\begin{flalign}
\hspace{6mm}d_{\mathcal{G} \setminus S'}(p,q) &\leq (4+5\epsilon) \cdot d_{l,w}({p_l}',{q_l}')&&\nonumber\\
        &\leq (4+5\epsilon)(1 + \epsilon) \cdot d_{\pi,w}(p,q)&&\nonumber\\
        &\text{[from} \ (\ref{eq26})\text{]}&&\nonumber\\
        &\leq (4+14\epsilon) \cdot d_{\pi,w}(p,q).&\nonumber
\end{flalign}}

\noindent
Thus, $\mathcal{G}$ is a $(k,(4+\epsilon))$-vertex fault-tolerant additive weighted geodesic spanner for the set $S$ of points located in the simple polygon $P$. 
\end{proof}

\section{Vertex fault-tolerant additive weighted spanner for points in a polygonal domain}
\label{sect:polydom}

We devise an algorithm to compute a geodesic $(k,(4+\epsilon))$-vertex fault-tolerant spanner for a set $S$ of $n$ points lying in the free space $\cal{D}$ of the given polygonal domain $\cal{P}$ while each input point is associated with a non-negative weight.
The polygonal domain $\cal{P}$ consists of a simple polygon and $h$ simple polygonal holes located interior to to that polygon. 
Our algorithm depends on the algorithm given in \cite{conf/compgeom/AbamAHA15} to compute a $(5+\epsilon)$-spanner for a set of unweighted points lying in $\mathcal{D}$. 
We decompose the free space $\mathcal{D}$ into simple polygons using $O(h)$ splitting segments such that no splitting segment crosses any of the holes of $\cal{D}$ and each of the resultant simple polygons has at most three splitting segments bounding it.
As part of this decomposition, two vertical line segments are drawn, one upwards and the other downwards, respectively from the leftmost and rightmost extreme (along the $x$-axis) vertices of each hole. 
If any of the resulting simple polygons has more than three splitting segments on its boundary, then that simple polygon is further decomposed. 
To achieve efficiency, a splitting segment is chosen such that it has around half of its bounding splitting segments on either of its sides.
This algorithm results in partitioning $\mathcal{D}$ into $O(h)$ simple polygons.
Further, a graph ${\mathcal{G}}_{d}$ is constructed where each vertex of ${\mathcal{G}}_{d}$ corresponds to a simple polygon of this decomposition. 
Each vertex $v$ of ${\mathcal{G}}_{d}$ is associated with a weight equal to the number of points that lie inside the simple polygon corresponding to $v$. 
Two vertices are connected by an edge in $\mathcal{G}_d$ whenever their corresponding simple polygons are adjacent to each other in the decomposition. 
We note that $\mathcal{G}_d$ is a planar graph.
Hence, we use the following theorem from \cite{journals/siamdm/AlonST94} to compute a $O(\sqrt{h})$-separator $R$ of $\mathcal{G}_d$.

\begin{theorem}
\label{thm4}
Suppose $G=(V,E)$ is a planar vertex-weighted graph with $|V| = m$. 
Then, an $O(\sqrt{m})$-separator for $G$ can be computed in $O(m)$ time. That is, $V$ can be partitioned into sets $P, Q$ and $R$ such that $|R| = O(\sqrt{m})$, there is no edge between $P$ and $Q$, and $w(P),w(Q) \leq \frac{2}{3}w(V)$.
Here, $w(X)$ is the sum of weights of all vertices in $X$.
\end{theorem}

We compute a $O(\sqrt{h})$-separator $R$ for the graph ${\cal G}_d$ using Theorem~\ref{thm4}.
Let $P, Q,$ and $R$ be the sets into which the vertices of ${\cal G}_d$ is partitioned.
For each vertex $r \in R$, we collect the bounding splitting segments of the simple polygon corresponding to $r$ into $H$ i.e., $O(\sqrt{h})$ splitting segments are collected into a set $H$. 
For each splitting segment $l$ in $H$, we proceed as follows. 
For each point $p$ that lies in the given simple polygon, we find the projection $p_l$ of $p$ on $l$; we assign the weight $w(p) + d_\pi(p, p_l)$ to point $p_l$ and include $p_l$ into the set $S_l$ corresponding to points projected on to line $l$.
We compute the $(k, 4+\epsilon)$-VFTAWS ${\cal G}_l$ for the set $S_l$ of points using the algorithm given in Section~\ref{sect:rd} for additive weighted points located in $\mathbb{R}^d$. 
For every edge $(r, s)$ in ${\cal G}_l$, we introduce an edge $(p, q)$ in $G$, where $r$ (resp. $s$) is the projection of $p$ (resp. $q$) on $l$. 
Recursively, we compute vertex-fault tolerant additive weighted spanner for points lying in simple polygon corresponding to $P$ (resp. $Q$). 
The recursion is continued till $P$ (resp. $Q$) contains exactly one vertex.
We first prove that this algorithm computes a $(k, (12+\epsilon))$-vertex fault-tolerant spanner.
Further, we modify this algorithm to compute a $(k, (4+\epsilon))$-vertex fault-tolerant spanner. 

\begin{algorithm}

    \SetKwInOut{Input}{Input}
    \SetKwInOut{Output}{Output}
    
    \Input{The polygonal domain $\mathcal{D}$, the set $S$ on $n$ points, non-negative weights associated to points in $S$, integer $k \geq 1$ and a real number $0 < \epsilon \leq 1$}
    
    \Output{$(k,(12+\epsilon))$-VFTS $\mathcal{G}$}
    
    \caption{$k$-\large A\small DDITIVE\large P\small OLYDOM\large C\small ULST\large FTS($S, \mathcal{D}, k, \epsilon$)}
    \label{alg:addpolyclustdomfts}

    \begin{algorithmic}[1]
        
    \STATE decompose $\mathcal{D}$ into $O(h)$ simple polygons such that each of these simple polygons have at most three splitting lines on its boundary
    
	    \STATE construct a planar graph ${\mathcal{G}}_{d}$ 
	    \scriptsize{} (description of ${\mathcal{G}}_{d}$ is in text) \normalsize{}
    
    \STATE initialize the set $X$ to contain all the vertices of $\mathcal{G}_d$
    
    \WHILE{$card(X) \geq 1$}
    
        \IF{$card(X) = 1$}
        
            \STATE $k$-\large A\small DDITIVE\large P\small OLYGON\large C\small ULST\large FTS($S \cap simplepoly(X), simplepoly(X), k,\epsilon$)
            
            \small // $simplepoly(X)$ denotes the simple polygon corresponding to the only vertex in $X$

        \ELSE
        
            \STATE compute a $O(\sqrt{h})$-separator $R$ for the graph $\mathcal{G}_d$ using Theorem \ref{thm4}; \\
	    let $P,Q$ and $R$ be the sets into which the vertices of $\mathcal{G}_d$ is partitioned \\
            
            \STATE for every $r \in R$, add the bounding splitting segments of the simple polygon corresponding \\
	    to $r$ to a set $H$
            
            \FOR{every $l \in H$}
        
                \FOR{every $p \in \mathcal{D} \cap S$}
            
                    \STATE find the projection $p_l$ of $p$ on $l$
                    
                    \STATE assign weight equal to $w(p) + d_{\pi}(p,p_l)$ to $p_l$
            
                    \STATE $S_l := S_l \cup \{p_l\}$
                
                \ENDFOR
        
                \STATE compute the $(k, (4+\epsilon))$-vertex fault-tolerant spanner for the set $S_l$ using Algorithm 1
        
                \STATE for every edge $(r,s)$ in $\mathcal{G}_l$, add the edge $(p,q)$ to $\mathcal{G}$, where\\
		$r$ (resp. $s$) is the projection of $p$(resp. $q$) on $l$
                
            \ENDFOR
        
            \STATE 
            
            $k$-\large A\small DDITIVE\large P\small OLYDOM\large C\small ULST\large FTS($S \cap unionpoly(P),\mathcal{D} \cap unionpoly(P),k,\epsilon$)
            
            \small // $unionpoly(P)$ denotes the set of simple polygons corresponding to each vertex of $P$
    
            \STATE 
            
            $k$-\large A\small DDITIVE\large P\small OLYDOM\large C\small ULST\large FTS($S \cap unionpoly(Q),\mathcal{D} \cap unionpoly(Q),k,\epsilon$)
            
            \small // $unionpoly(Q)$ denotes the set of simple polygons corresponding to each vertex of $Q$
        
        \ENDIF
            
    \ENDWHILE

    \end{algorithmic}
    \label{algo:polydom}

\end{algorithm}

\begin{lemma}
\label{lem20}
The spanner $G$ is a geodesic $(k, (12 + 15\epsilon))$-vertex fault-tolerant additive weighted spanner for points in $\cal{D}$.
\end{lemma}
\begin{proof}
Using induction on the number of points, we show that there exists a $(12 + 15\epsilon)$-spanner path between $p$ and $q$ in $\mathcal{G} \setminus S^{'}$.
The induction hypothesis assumes that for the number of points $k' < |S|$, there exists a $(12 + 15\epsilon)$-spanner path between any two points belonging to $G \ S$.
Consider a set $S^{'} \subset S$ such that $|S^{'}| \leq k$ and two arbitrary points $p$ and $q$ from the set $S \setminus S^{'}$. 
Here, as described above, $P, Q$, and $R$ correspond to vertices of a planar graph ${\cal G}_d$.
The union of simple polygons that correspond to vertices of $P$ (resp. $Q, R$) is denoted with $poly(P)$ (resp. $poly(Q), poly(R)$).
Also, the set $H$ is as described in the algorithm.
Based on the location of $p$ and $q$, the following cases arise:  
(i) both $p$ and $q$ are lying in $P' \in \{poly(P), poly(Q)$, and $poly(R)\}$ and the geodesic Euclidean shortest path between $p$ and $q$ does not intersect any splitting segment from the set $H$, and
(ii) $p$ is lying in $P' \in \{poly(P), poly(Q), poly(R)\}$ and $q$ is lying in $P'' \in \{poly(P), poly(Q), poly(R)\}$ where $P' \ne P''$.
In case (i), if $P'$ is a simple polygon, then we can apply algorithm for simple polygons from Section~\ref{sect:simppoly} and obtain a $(4+14\epsilon)$-path between $p$ and $q$.
Otherwise, from the induction hypothesis, there exists a $(12 + 15\epsilon)$-path between $p$ and $q$.
In case (ii), a shortest path from $p$ and $q$ intersects at least one of the $O(\sqrt{h})$ splitting segments in $H$, say $l$. 
Let $\pi(p, q)$ be a shortest path between $p$ and $q$ that intersects $l$ at some point. 
Let $r$ be this point of intersection.
Since ${\mathcal{G}}_l$ is a $(k, (4+5\epsilon))$-VFTAWS, there exists a path $P'$ between $p_l$ and $q_l$ in ${\mathcal{G}}_l$ with length at most $(4+5\epsilon) d_{l,w}(p_l, q_l)$.
By replacing each vertex $x_l$ of $P'$ by $x \in S$ such that $x_l$ is the projection of $x$ on $l$, gives a path between $p$ and $q$ in $\mathcal{G} \setminus S^{'}$. 
Thus, the length of the path $d_{\mathcal{G} \setminus S^{'}}(p,q)$ is less than or equal to the length of the path $P'$ in $\mathcal{G}_l$.
For every $x,y \in S$,

{\setlength{\abovedisplayskip}{0pt}
\begin{flalign}
\hspace{6mm}d_{\pi,w}(x,y) &= w(x) + d_{\pi}(x,y) + w(y)&&\nonumber\\
        &\leq w(x) + d_{\pi}(x,x_l) + d_{\pi}(x_l,y_l) + d_{\pi}(y_l,y) + w(y)&&\nonumber\\
        &\text{[by the triangle inequality]}&&\nonumber\\
        &= w(x_l) + d_{\pi}(x_l,y_l) + w(y_l)&&\nonumber\\
    &\text{[since the weight associated with projection} \ z_l \ \text{of every point} \ z \ \text{is}&&\nonumber\\
    &w(z) + d_{\pi}(z,z_l)\text{]}&&\nonumber\\
        &=d_{l,w}(x_l,y_l)\label{eq39}&
\end{flalign}}

This implies, 
{\setlength{\abovedisplayskip}{0pt}
\begin{flalign}
\hspace{6mm}d_{\mathcal{G} \setminus S^{'}}(p,q) &= \sum_{x_l,y_l \in P} d_{\pi,w}(x,y)&&\nonumber\\
        &\leq \sum_{x_l,y_l \in P} d_{l,w}(x_l,y_l)&&\nonumber\\
        &\text{[\text{from} \ (\ref{eq39})]}&&\nonumber\\
        &\leq (4+5\epsilon).d_{l,w}(p_l,q_l)\label{eq40}&&\\
	&\text{[since} \ {\mathcal{G}}_l \ \text{is a geodesic} \ (k,(4+5\epsilon))\text{-VFTAWS]}&\nonumber\\
    &= (4+5\epsilon) \cdot [w(p_l) + d_l(p_l,q_l) + w(q_l)]&&\nonumber\\
        &= (4+5\epsilon) \cdot [w(p_l) + d_{\pi}(p_l,q_l) + w(q_l)]&&\nonumber\\
    &\text{[since} \ P \ \text{contains $l$, shortest path between } p_l \text{ and } q_l \text{ along } l&&\nonumber\\
        &\text{is same as the geodesic shortest path between } p_l \text{ and } q_l]&&\nonumber\\
        &= (4+5\epsilon) \cdot [w(p) + d_{\pi}(p,p_l) + d_{\pi}(p_l,q_l) + d_{\pi}(q_l,q) + w(q)]\label{eq20}&&\\
    &\text{[since the weight associated with projection} \ z_l \ \text{of every point} \ z \ \text{is}&&\nonumber\\
    &w(z) + d_{\pi}(z,z_l) \text{]}.&\nonumber
\end{flalign}}

\noindent
Since $r$ is a point belonging to both $l$ as well as to $\pi(p, q)$,
{\setlength{\abovedisplayskip}{0pt}
\begin{flalign}
\hspace{6mm}&d_{\pi}(p,p_l) \leq d_{\pi}(p,r) \ \text{and} \ d_{\pi}(q,q_l) \leq d_{\pi}(q,r).\label{eq21}&
\end{flalign}}

\noindent
Substituting (\ref{eq21}) into (\ref{eq20}),
{\setlength{\abovedisplayskip}{0pt}
\begin{flalign}
\hspace{6mm}d_{\mathcal{G} \setminus S'}(p,q) &\leq (4+5\epsilon) \cdot [w(p) + d_{\pi}(p,r) + d_{\pi}(p_l,q_l) + d_{\pi}(r,q) + w(q)]&&\nonumber\\
        &\leq (4+5\epsilon) \cdot [w(p) + d_{\pi}(p,r) + w(r) + d_{\pi}(p_l,q_l) + w(r) + d_{\pi}(r,q) + w(q)]&&\nonumber\\
        &\text{[since the weight associated with every point is non-negative]}&&\nonumber\\
        &= (4+5\epsilon) \cdot [d_{\pi,w}(p,r) + d_{\pi}(p_l,q_l) + d_{\pi,w}(r,q)]&&\nonumber\\
        &= (4+5\epsilon) \cdot [d_{\pi,w}(p,q) + d_{\pi}(p_l,q_l)]\label{eq22}&&\\
        &\text{[since} \ \pi(p,q) \ \text{intersects} \ l \text{ at} \ r\text{, by optimal substructure property of shortest}&&\nonumber\\ 
        &\text{paths,} \ \pi(p,q) = \pi(p,r) + \pi(r,q)\text{]}.&\nonumber
\end{flalign}}

\noindent
Consider
{\setlength{\abovedisplayskip}{0pt}
\begin{flalign}
\hspace{6mm}d_{\pi}(p_l,q_l) &\leq d_{\pi}(p_l,p) + d_{\pi}(p,q) + d_{\pi}(q,q_l)&&\nonumber\\
        &\text{[since} \ \pi \ \text{follows triangle inequality]}&&\nonumber\\
        &\leq d_{\pi}(r,p) + d_{\pi}(p,q) + d_{\pi}(q,r)&&\nonumber\\
        &\text{[using} \ (\ref{eq21})\text{]}&&\nonumber\\
        &\leq w(r) + d_{\pi}(r,p) + w(p) + w(p) + d_{\pi}(p,q) + w(q) + w(q) + d_{\pi}(q,r) + w(r)&&\nonumber\\
        &\text{[since weight associated with every point is non-negative]}&&\nonumber\\
        &= d_{\pi,w}(p,r) + d_{\pi,w}(p,q) + d_{\pi,w}(r,q)&&\nonumber\\
        &= d_{\pi,w}(p,q) + d_{\pi,w}(p,q)&&\nonumber\\
        &\text{[since} \ \pi(p,q) \ \text{intersects} \ l \text{ at } \ r\text{, by optimal substructure property of shortest}&&\nonumber\\ 
        &\text{paths,} \ \pi(p,q) = \pi(p,r) + \pi(r,q)\text{]}&&\nonumber\\
        &= 2d_{\pi,w}(p,q).\label{eq23}&
\end{flalign}}
Substituting (\ref{eq23}) into (\ref{eq22}), $d_{G \setminus S'}(p,q) \leq 3(4+5\epsilon) \cdot d_{\pi,w}(p,q)$.
\end{proof}

We further improve the stretch factor of $\mathcal{G}$ by applying the refinement given in \cite{conf/soda/AbamBS17} to the above-described algorithm.
In doing this, for each point $p \in S$, we compute the geodesic projection $p_{\gamma}$ of $p$ on the splitting line $\gamma$ and we construct a set $S(p,\gamma)$ as defined herewith. 
Let $\gamma(p) \subseteq \gamma$ be $\{ x \in \gamma : d_{\gamma,w}(p_{\gamma},x) \leq (1 + 2\epsilon) \cdot d_{\pi}(p,p_{\gamma}) \}$. 
Here, for any $p,q \in S$, $d_{\gamma ,w}(p, q)$ is equal to $0$ if $p = q$; otherwise, it is equal to $w(p) + d_\gamma(p,q) + w(q)$.
We divide $\gamma(p)$ into $c$ pieces with $c \in O(1/\epsilon^{2})$: each piece is denoted by $\gamma_{j}(p)$ for $1 \leq j \leq c$.
For every piece $j$, we compute the point $p_\gamma^{(j)}$ nearest to $p$ in $\gamma_j(p)$.
The set $S(p,\gamma)$ is defined as $\{ p_\gamma^{(j)} : p_\gamma^{(j)} \in \gamma_{j}(p)$  and $1 \leq j \leq c\}$. 
For every $r \in S(p,\gamma)$, the non-negative weight $w(r)$ of $r$ is set to $w(p) + d_{\pi}(p,r)$.
Let $S_\gamma$ be $\cup_{p \in S}S(p,\gamma)$.
We replace the set $S_l$ in computing $G$ with the set $S_\gamma$ and compute a $(k,(4+5\epsilon))$-VFTAWS $\mathcal{G}_l$ for the set $S_l$ using the algorithm for points in $\mathbb{R}^d$ given in Section~\ref{sect:rd}.
Further, for every edge $(r,s)$ in $\mathcal{G}_l$, we add the edge $(p,q)$ to $G$ such that $r \in S(p,l)$ and $s \in S(q,l)$.
The rest of the algorithm remains the same.

\begin{theorem}
\label{thm:polydom}
Let $S$ be a set of $n$ points in a polygonal domain $\cal D$ with non-negative weights associated to points via weight function $w$.
For any fixed constant $\epsilon > 0$, there exists a $(k, (4+\epsilon))$-vertex fault-tolerant additive weighted geodesic spanner with $O(\frac{k n \sqrt{h}}{\epsilon^{2}} \lg n)$ edges for the metric space $(S, d_{\pi, w})$.
\end{theorem}

\begin{proof}
Let $S(n)$ be the size of $\mathcal{G}$. 
Our algorithm adds $O(\frac{k n \sqrt{h}}{\epsilon^{2}})$ edges at each recursive level except for the last level.
At every leaf node $l$ of the recurrence tree, we add $O(\frac{k n_x}{\epsilon^{2}} \lg n_x)$ edges, where $n_{x}$ is the number of points inside the simple polygon corresponding to $l$. 
Hence, the number of edges of $G$ is $O(\frac{k n \sqrt{h}}{\epsilon^{2}} \lg n)$.

Next, we prove the stretch factor of the spanner.
Consider any set $S^{'} \subset S$ such that $|S^{'}| \leq k$ and two arbitrary points $p$ and $q$ from the set $S \setminus S^{'}$. 
We show that there exists a $(4 + 14\epsilon)$-spanner path between $p$ and $q$ in $\mathcal{G} \setminus S^{'}$.  
If $r \notin l(p)$, then we set $p_l'$ (resp. $q_l'$) equal to $p_l$ (resp. $q_l$). 
Otherwise, $p_l'$ (resp. $q_l'$) is set as the point from $S(p,l)$ (resp. $S(q,l)$) that is nearest to $p$ (resp. $q$).
(The $r$ is defined before the theorem statement.) 

{\setlength{\abovedisplayskip}{0pt}
\begin{flalign}
\hspace{6mm}d_{l,w}(p_l',q_l') &= w(p_l') + d_l(p_l', q_l') + w(q_l')&&\nonumber\\
        &\leq w(p_l') + d_l(p_l',r) + d_l(r,q_l') + w(q_l')&&\nonumber\\
        &\text{[by the triangle inequality]}&&\nonumber\\
        &\leq w(p_l') + d_l(p_l',r) + w(r) + w(r) +  d_l(r,q_l') + w(q_l')&&\nonumber\\
        &\text{[since the weight associated with each point is non-negative]}&&\nonumber\\
    &= w(p) + d_{\pi}(p,p_l') + d_l(p_l',r) + w(r) + w(r) + d_l(r,q_l') +&\nonumber\\
    &d_{\pi}(q_l',q) + w(q)\label{eq41}&&\\
        &\text{[due to the assignment of the weight to the projection of any point]}&\nonumber
\end{flalign}}

\noindent
From the triangle inequality, we know the following: 
{\setlength{\abovedisplayskip}{0pt}
\begin{flalign}
\hspace{6mm}&d_{\pi}(p,p_l') + d_l(p_l',r) \leq d_{\pi}(p,r)\text{, and} \label{eq25'}\\
\hspace{6mm}&d_l(r,q_l') + d_{\pi}(q_l',q) \leq d_{\pi}(r,q).\label{eq25}&
\end{flalign}}

\noindent
Substituting (\ref{eq25'}) and (\ref{eq25}) in (\ref{eq41}),
{\setlength{\abovedisplayskip}{0pt}
\begin{flalign}
\hspace{6mm}d_{l,w}(p_l',q_l') &\leq w(p) + d_\pi(p,r) + w(r) + w(r) + d_\pi(r,q) + w(q)&&\nonumber\\
        &= d_{\pi,w}(p,r) + d_{\pi,w}(r,q)&&\nonumber\\
        &= d_{\pi,w}(p,q)\label{eq26}&&\\
        &[\text{since } r \in l \cap \pi(p,q), \text{ by the optimal substructure property of shortest}&&\nonumber\\ 
        &\text{paths,} \ \pi(p,q) = \pi(p,r) + \pi(r,q)].&\nonumber
\end{flalign}}

Replacing $p_l$ (resp. $q_l$) by $p_l'$ (resp. $q_l'$) in inequality (\ref{eq40}),
{\setlength{\abovedisplayskip}{0pt}
\begin{flalign}
\hspace{6mm}d_{\mathcal{G} \setminus S^{'}}(p,q) &\leq (4+5\epsilon).d_{l,w}(p_l',q_l')&&\nonumber\\
	&\leq (4+5\epsilon)d_{\pi,w}(p,q) \hspace{0.2in} \text{[from} \ (\ref{eq26})\text{].}&&\nonumber
\end{flalign}}
Thus, $G$ is a geodesic $(k, (4 + \epsilon))$-VFTAWS for $S$. 
\end{proof}

\section{Vertex fault-tolerant additive weighted spanner for points on a terrain}
\label{sect:terrains}

In this section, we present an algorithm to compute a geodesic $(k,(4+\epsilon))$-VFTAWS with $O(\frac{kn\lg{n}}{\epsilon^2})$ edges for any given set $S$ of $n$ non-negative weighted points lying on a polyhedral terrain $\mathcal{T}$. 
We denote the boundary of $\mathcal{T}$ with $\partial \mathcal{T}$. 
The following distance function $d_{\mathcal{T},w}: S \times S \rightarrow \rm I\!R \cup \{0\}$ is used to compute the geodesic distance on $\mathcal{T}$ between any two points $p, q \in S$: $d_{\mathcal{T},w}(p,q) = w(p) + d_{\mathcal{T}}(p,q) + w(q)$.
Here, $w(p)$ (resp. $w(q)$) is the non-negative weight of $p \in S$ (resp. $q \in S$).
We denote a geodesic Euclidean shortest path between any two points $a$ and $b$ on $\mathcal{T}$ with $\pi(a, b)$.
For any two points $x, y \in \partial \mathcal{T}$, we denote the closed region lying to the right (resp. left) of $\pi(x,y)$ when going from $x$ to $y$, including the points lying on the shortest path $\pi(x,y)$ with $\pi^{+}(x,y)$ (resp. $\pi^{-}(x,y)$).
The {\it projection} $p_\pi$ of a point $p$ on the shortest path $\pi$ between two points lying on the polyhedral terrain $\mathcal{T}$ is defined as a point on $\pi$ that is at the minimum geodesic distance from $p$ among all the points located on $\pi$.
For three points $u,v,w \in \mathcal{T}$, the closed region bounded by shortest paths $\pi(u,v)$, $\pi(v,w)$, and $\pi(w,u)$ is termed {\it sp-triangle}, denoted with $\Delta(u, v, w)$.
If the points $u, v, w \in \mathcal{T}$ are clear from the context, we denote the sp-triangle with $\Delta$.
In the following, we restate a Theorem from \cite{conf/soda/AbamBS17}, which is useful for our analysis.

\begin{theorem}
\label{thm7}
For any set $P$ of $n$ points on a polyhedral terrain $\mathcal{T}$, there exists a balanced sp-separator: a shortest path $\pi(u,v)$ connecting two points $u,v \in \partial \mathcal{T}$ such that $\frac{2n}{9} \leq |\pi^{+}(u,v) \cap P| \leq \frac{2n}{3}$, or a sp-triangle $\Delta$ such that $\frac{2n}{9} \leq |\Delta \cap P| \leq \frac{2n}{3}$.
\end{theorem}

Thus, an sp-separator is either bounded by a shortest path (in the former case) or by three shortest paths (in the latter case).
Let $\gamma$ be a shortest path that belongs to an sp-separator.
First, a balanced sp-separator as given in Theorem~\ref{thm7} is computed. 
The sets $S_{in}$ and $S_{out}$ comprising of points are defined as follows: if the sp-separator is a shortest path then define $S_{in}$ to be $\gamma^{+}(u,v) \cap S$; otherwise, $S_{in}$ is $\Delta \cap S$; points in $S$ that do not belong to $S_{in}$ are in $S_{out}$.
For each $p \in S$, we compute the projection $p_{\gamma}$ of $p$ on every shortest path $\gamma$ of sp-separator, and associate a weight $d_{\mathcal{T}}(p, p_{\gamma})$ with $p_{\gamma}$.
Let $S_{\gamma}$ be a set defined as $\cup_{p \in S}\hspace{0.02in} p_{\gamma}$.
Our algorithm computes a $(2+\epsilon)$-spanner ${\mathcal{G}}_{\gamma}$ for the weighted points in $S_{\gamma}$.
Further, for each edge $(p_{\gamma}, q_{\gamma})$ in ${\mathcal{G}}_{\gamma}$, an edge $(p,q)$ is added to $\mathcal{G}$, where $p_{\gamma}$ (resp. $q_{\gamma}$) is the projection of $p$ (resp. $q$) on $\gamma$. 
The spanners for the sets $S_{in}$ and $S_{out}$ are computed recursively, and the edges from these spanners are added to $\mathcal{G}$.
In the base case, if $|S| \leq 3$ then a complete graph on the set $S$ is constructed. 
We first obtain a $(k, (12+15\epsilon))$-vertex fault-tolerant additive weighted spanner for the set $S$ of points lying on the terrain $\mathcal{T}$.
(This construction is later modified to compute a $(k, (4+\epsilon))$-VFTAWS.)
In specific, with every projected point $p_{\gamma}$, instead of associating $d_{\mathcal{T}}(p, p_{\gamma})$ as the weight of $p_{\gamma}$, we associate $w(p)+d_{\mathcal{T}}(p, p_{\gamma})$ as the weight of $p_{\gamma}$.
The rest of the algorithm in constructing $\cal{G}$ remains the same as in \cite{conf/soda/AbamBS17}.

\begin{algorithm}[H]

    \SetKwInOut{Input}{Input}
    \SetKwInOut{Output}{Output}
    
    \Input{The polyhedral terrain $\mathcal{T}$, the set $S$ on $n$ additive weighted points, integer $k \geq 1$ and a real number $0 < \epsilon \leq 1$}
    
    \Output{$(k,(12+\epsilon))$-VFTS $\mathcal{G}$}
    
    \caption{$k$-\large A\small DDITIVE\large T\small ERRAIN\large C\small LUST\large FTS($S, \mathcal{D}, k, \epsilon$)}
    \label{alg:addterclustfts}
    
    \begin{algorithmic}[1]
    
    \WHILE{$card(\mathcal{T} \cap S) \geq 1$}
        
        \STATE compute a balanced sp-separator for $\mathcal{T}$ using the algorithm given in \cite{conf/soda/AbamBS17}
        
        \STATE if the balanced sp-separator is a shortest path $\pi(u,v)$ between $u,v \in \partial \mathcal{T}$, then define set $S_{in}$ to be $\pi^{+}(u,v) \cap S$; otherwise, define set $S_{in}$ to be $\Delta \cap S$, where $\Delta$ is a sp-triangle.
        
        \STATE further, define the set $S_{out}$ to be $S \setminus S_{in}$
        
        \FOR{every bounding shortest path $\gamma$ of the sp-separator}
    
            \STATE initialize the set $S_\gamma$ of the projections to $\phi$
    
            \FOR{every $p \in S$}
    
                \STATE find the projection $p_\gamma$ of $p$ on $\gamma$
            
                \STATE assign a weight equal to $w(p) + d_{\mathcal{T}}(p,p_\gamma)$ to $p_\gamma$
        
                \STATE $S_\gamma := S_\gamma \cup \{p_\gamma\}$
        
            \ENDFOR
    
            \STATE compute a $(k,(4+\epsilon))$-vertex fault-tolerant spanner $\mathcal{G}_\gamma$ for the set $S_\gamma$ using Algorithm 1
    
            \STATE for every edge $(r,s)$ in $\mathcal{G}_\gamma$, add the edge $(p,q)$ to $\mathcal{G}$ where $r$(resp. $s$) is the projection of $p$ \\
	    (resp. $q$) on $l$
        
        \ENDFOR
    
        \STATE $k$-\large A\small DDITIVE\large T\small ERRAIN\large C\small LUST\large FTS($S_{in}, X, k, \epsilon$)
        
        \small //$X$ set to is $\pi^{+}(u,v)$ if the balanced sp-separator is a shortest path $\pi(u,v)$; else it is set to $\Delta$
    
        \STATE $k$-\large A\small DDITIVE\large T\small ERRAIN\large FTS($S_{out}, Y, k, \epsilon$)
        
        \small // $Y$ set to is $\mathcal{T} \setminus \pi^{+}(u,v)$ if the balanced sp-separator is a shortest path $\pi(u,v)$; else it is set to $\mathcal{T} \setminus \Delta$
        
    \ENDWHILE
    \end{algorithmic}

\end{algorithm}

To prove the graph $\mathcal{G}$ is a geodesic $(k,(12+15\epsilon))$-VFTAWS for the points in $S$, we use induction on the number of points. 
Consider any set $S' \subset S$ such that $|S'| \leq k$ and two arbitrary points $p$ and $q$ from the set $S \setminus S'$.
We show that there exists a path between $p$ and $q$ in ${\cal G} \setminus S'$ such that $d_{\cal G}(p, q)$ is at most $(12+15\epsilon)d_{\pi,w}(p,q)$.
The induction hypothesis assumes that for the number of points $k' < |S|$ in a region of $\mathcal{T}$, there exists a $(12+15\epsilon)$-spanner path between any two points belonging to the given region in $\mathcal{G} \setminus S'$. 
As part of the inductive step, we extend it to $n$ points.
For the case of both $p$ and $q$ are on the same side of a bounding shortest path $\gamma$ of the balanced separator, i.e., both are in $S_{in}$ or $S_{out}$, by induction hypothesis (as the number of points in $S_{in}$ or $S_{out}$ is less than $|S|$), there exists a $(12+15\epsilon)$-spanner path between $p$ and $q$ in $\mathcal{G} \setminus S'$. 
The only case remains to be proved is when $p$ lies on one side of $\gamma$ and $q$ lies on the other side of $\gamma$, i.e., $p \in S_{in}$ and $q \in S_{out}$ or, $q \in S_{in}$ and $p \in S_{out}$.
W.l.o.g., we assume that the former holds. 
Let $r$ be a point on $\gamma$ at which the geodesic shortest path $\pi(p,q)$ between $p$ and $q$ intersects $\gamma$.
Since ${\mathcal{G}}_\gamma$ is a $(k, (4+5\epsilon))$-VFTS, there exists a path $P$ between $p_\gamma$ and $q_\gamma$ in ${\mathcal{G}}_\gamma$ of length at most $(4+5\epsilon).d_{\gamma,w}(p_\gamma, q_\gamma)$. 
Let $P'$ be the path obtained by replacing each vertex $x_\gamma$ of $P$ by $x \in S$ such that $x_\gamma$ is the projection of $x$ on $\gamma$. 
Note that the path $P'$ is between nodes $p$ and $q$ in $\mathcal{G} \setminus S'$.
The length $d_{\mathcal{G} \setminus S'}(p,q)$ of path $P'$ is less than or equal to the length of the path $P$ in $\mathcal{G}_\gamma$.
In the following, we show that $d_{\mathcal{G} \setminus S'}(p,q) \le (12+15\epsilon)d_{{\mathcal T}, w}(p,q)$.
\hfil\break

\noindent For every $x,y \in S$,
{\setlength{\abovedisplayskip}{0pt}
\begin{flalign}
\hspace{6mm}d_{\mathcal{T},w}(x,y) &= w(x) + d_{\mathcal{T}}(x,y) + w(y)&&\nonumber\\
        &\leq w(x) + d_{\mathcal{T}}(x,x_{\gamma}) + d_{\mathcal{T}}(x_{\gamma},y_{\gamma}) + d_{\mathcal{T}}(y_{\gamma},y) + w(y)&&\nonumber\\
        &\text{[by the triangle inequality]}&&\nonumber\\
        &= w(x_{\gamma}) + d_{\mathcal{T}}(x_{\gamma},y_{\gamma}) + w(y_{\gamma})&&\nonumber\\
    &\text{[since the weight associated with projection} \ z_{\gamma} \ \text{of every point} \ z \ \text{is }&&\nonumber\\
    &w(z) + d_{\mathcal{T}}(z,z_{\gamma})\text{]}&&\nonumber\\
        &=d_{\gamma,w}(x_{\gamma},y_{\gamma}).\label{eq43}&
\end{flalign}}

\noindent
This implies,
{\setlength{\abovedisplayskip}{0pt}
\begin{flalign}
\hspace{6mm}d_{\mathcal{G} \setminus S'}(p,q) &= \sum_{x_\gamma,y_\gamma \in P} d_{\mathcal{T},w}(x,y)&&\nonumber\\
        &\leq \sum_{x_\gamma,y_\gamma \in P} d_{\gamma,w}(x_\gamma,y_\gamma)&&\nonumber\\
        &\text{[from} \ (\ref{eq43})\text{]}&&\nonumber\\
        &\leq (4+5\epsilon).d_{\gamma,w}(p_{\gamma},q_{\gamma})\label{eq44}&&\\
        &\text{[since} \ {\mathcal{G}}_{\gamma} \ \text{is a} \ (k,(4+5\epsilon)) \text{-vertex fault tolerant geodesic spanner]}&&\nonumber\\
        &= (4+5\epsilon).[w(p_{\gamma}) + d_{\mathcal{T}}(p_{\gamma},q_{\gamma}) + w(q_{\gamma})]&&\nonumber\\
        &\text{[since} \ \gamma \ \text{is a shortest path on} \ \mathcal{T}\text{, shortest path between any two}&&\nonumber\\
        &\text{points on} \ \gamma \ \text{is a geodesic shortest path on }\mathcal{T}]&&\nonumber\\
        &= (4+5\epsilon).[w(p) + d_{\mathcal{T}}(p,p_{\gamma}) + d_{\mathcal{T}}(p_{\gamma},q_{\gamma}) + d_{\mathcal{T}}(q_{\gamma},q) + w(q)]\label{eq45}&&\\
        &\text{[since the weight associated with projection} \ z_{\gamma} \ \text{is } \ w(z) + d_{\mathcal{T}}(z,z_{\gamma})\text{]}.&\nonumber
\end{flalign}}

\noindent
By the definition of projection of any point on $\gamma$, we know that
{\setlength{\abovedisplayskip}{0pt}
\begin{flalign}
\hspace{6mm}&d_{\mathcal{T}}(p,p_{\gamma}) \leq d_{\mathcal{T}}(p,r) \ \text{and} \ d_{\mathcal{T}}(q,q_{\gamma}) \leq d_{\mathcal{T}}(q,r).\label{eq46}&
\end{flalign}}

\noindent
Substituting (\ref{eq46}) into (\ref{eq45}),
{\setlength{\abovedisplayskip}{0pt}
\begin{flalign}
\hspace{6mm}d_{\mathcal{G} \setminus S'}(p,q) &\leq (4+5\epsilon).[w(p) + d_{\mathcal{T}}(p,r) + d_{\mathcal{T}}(p_{\gamma},q_{\gamma}) + d_{\mathcal{T}}(r,q) + w(q)]&&\nonumber\\
        &\leq (4+5\epsilon).[w(p) + d_{\mathcal{T}}(p,r) + w(r) + d_{\mathcal{T}}(p_{\gamma},q_{\gamma}) + w(r) + d_{\mathcal{T}}(r,q) + w(q)]&&\nonumber\\
        &\text{[since the weight of every point is non-negative]}&&\nonumber\\
        &= (4+5\epsilon).[d_{\mathcal{T},w}(p,r) + d_{\mathcal{T}}(p_{\gamma},q_{\gamma}) + d_{\mathcal{T},w}(r,q)]&&\nonumber\\
        &= (4+5\epsilon).[d_{\mathcal{T},w}(p,q) + d_{\mathcal{T}}(p_{\gamma},q_{\gamma})]\label{eq47}&&\\
    &\text{[since} \ \pi(p,q) \  \text{intersects} \ \gamma \ \text{at} \ r\text{]}&&\nonumber\\ 
        &\text{paths} \ \pi(p,q) = \pi(p,r) + \pi(r,q)\text{]}.&\nonumber
\end{flalign}}

\noindent
Further,
{\setlength{\abovedisplayskip}{0pt}
\begin{flalign}
\hspace{6mm}d_{\mathcal{T}}(p_{\gamma},q_{\gamma}) &\leq d_{\mathcal{T}}(p_{\gamma},p) + d_{\mathcal{T}}(p,q) + d_{\mathcal{T}}(q,q_{\gamma})&&\nonumber\\
        &\text{[by the triangle inequality]}&&\nonumber\\
        &\leq d_{\mathcal{T}}(r,p) + d_{\mathcal{T}}(p,q) + d_{\mathcal{T}}(q,r)&\nonumber\\
        &\text{[using} \ (\ref{eq46})\text{]}&&\nonumber\\
    &\leq w(r) + d_{\mathcal{T}}(r,p) + w(p) + w(p) + d_{\mathcal{T}}(p,q) + w(q) + w(q) + &&\nonumber\\
    &d_{\mathcal{T}}(q,r) + w(r)&\nonumber\\
        \hspace{26mm}&\text{[since the weight of every point is non-negative]}&&\nonumber\\
        &= d_{\mathcal{T},w}(p,r) + d_{\mathcal{T},w}(p,q) + d_{\mathcal{T},w}(r,q)&&\nonumber\\
        &= d_{\mathcal{T},w}(p,q) + d_{\mathcal{T},w}(p,q)&&\nonumber\\
    &\text{[since} \ \pi(p,q) \  \text{intersects} \ \gamma \ \text{at} \ r\text{]}&&\nonumber\\
        &= 2d_{\mathcal{T},w}(p,q).\label{eq48}&
\end{flalign}}
Substituting (\ref{eq48}) into (\ref{eq47}) yields, $d_{\mathcal{G} \setminus S'}(p,q) \leq 3(4+5\epsilon).d_{\mathcal{T},w}(p,q)$.

We improve the stretch factor of $\mathcal{G}$ by applying the same refinement as the one used in the algorithm in Section~\ref{sect:polydom}.
Again, we denote the graph resulted after applying that refinement with $\cal{G}$.

\begin{theorem}
\label{lem26}
Let $S$ be a set of $n$ weighted points on a polyhedral terrain $\cal{T}$ with non-negative weights associated to points via weight function $w$.
For any fixed constant $\epsilon > 0$, there exists a $(k, (4+\epsilon))$-vertex fault-tolerant additive weighted geodesic spanner with $O(\frac{k n}{\epsilon^{2}}\lg{n})$ edges.
\end{theorem}

\begin{proof}
The argument for the number of edges is same as in the proof of Theorem~\ref{thm:polydom}.
To prove that the graph $\mathcal{G}$ is a geodesic $(k,(4+14\epsilon))$-VFTAWS for the points in $S$, we use induction on $|S|$.
\ignore {
In specific, we show that for any set $S' \subset S$ with $|S'| \le k'$ and for any two points $p, q \in S \setminus S'$, there exists a $(4+14\epsilon)$-spanner path between $p$ and $q$ in $\cal{G} \setminus S'$. 
The induction hypothesis assumes that for the number of points $k' < |S|$ in a region of $\mathcal{T}$, there exists a $(4+14\epsilon)$-spanner path between any two points of the given region in $\mathcal{G} \setminus S^{'}$. 
For the case of both $p$ and $q$ on the same side of a separator (i.e., both $p$ and $q$ are either in $S_{in}$ or in $S_{out}$), by induction hypothesis there exists a $(4+14\epsilon)$-spanner path between $p$ and $q$ in $\mathcal{G} \setminus S'$. 
The case remaining to be proved is when $p$ and $q$ lie on distinct sides of $\gamma$, i.e., $p \in S_{in}$ and $q \in S_{out}$ or $q \in S_{in}$ and $p \in S_{out}$. 
}
W.l.o.g., we assume that $p \in S_{in}$ and $q \in S_{out}$. 
Let $r$ be the point at which the geodesic shortest path $\pi(p, q)$ between $p$ and $q$ intersects $r$.
Since ${\mathcal{G}}_\gamma$ is a $(k,(4+5\epsilon))$-VFTS, there exists a path $R$ between $p_\gamma$ and $q_\gamma$ in ${\mathcal{G}}_\gamma$ of length at most $(4+5\epsilon).d_{\gamma,w}(p_\gamma, q_\gamma)$. 
By replacing each vertex $x_\gamma$ of $R$ by $x \in S$ such that $x_\gamma$ is the projection of $x$ on $\gamma$, yields a path $R'$ between $p$ and $q$ in $\mathcal{G} \setminus S'$. 
The length $d_{\mathcal{G} \setminus S'}(p,q)$ of path $R'$ is less than or equal to the length of the path $R$ in $\mathcal{G}_\gamma$.
If $r \notin \gamma(p)$, point $p_\gamma'$ (resp. $q_\gamma'$) is set as $p_{\gamma}$ (resp. $q_{\gamma}$).
Otherwise, $p_\gamma'$ (resp. $q_\gamma'$) is set as the point in $S(p,\gamma)$ (resp. $S(q,\gamma)$) that is nearest to $p$ (resp. $q$). 

{\setlength{\abovedisplayskip}{0pt}
\begin{flalign}
\hspace{6mm}d_{\gamma,w}(p_\gamma',q_\gamma') &= w(p_\gamma') + d_{\gamma}(p_\gamma',q_\gamma') + w(q_\gamma')&\nonumber\\
        &\leq w(p_\gamma') + d_{\gamma}(p_\gamma',r) + d_{\gamma}(r,q_\gamma') + w(q_\gamma')&&\nonumber\\
        &\text{[by the triangle inequality]}&&\nonumber\\
        &\leq w(p_\gamma') + d_{\gamma}(p_\gamma',r) + w(r) + w(r) +  d_{\gamma}(r,q_\gamma') + w(q_\gamma')&&\nonumber\\
        &\text{[since the weight of each point is non-negative]}&&\nonumber\\
    &= w(p) + d_{\mathcal{T}}(p,p_\gamma') + d_{\gamma}(p_\gamma',r) + w(r) + w(r) + d_{\gamma}(r,q_\gamma') +&&\nonumber\\
    &d_{\mathcal{T}}(q_\gamma',q) + w(q)\label{eq49}&&\\
        &\text{[due to the association of weight to the projections of points]}.&&\nonumber
\end{flalign}}

\noindent
From the triangle inequality, we know that 
$d_{\mathcal{T}}(p,p_\gamma') + d_{\gamma}(p_\gamma',r) \le d_{\mathcal{T}}(p,r)$, and
$d_{\gamma}(r,q_\gamma') + d_{\mathcal{T}}(q_\gamma',q) \le d_{\mathcal{T}}(r,q)$.
Hence (\ref{eq49}) is written as

{\setlength{\abovedisplayskip}{0pt}
\begin{flalign}
\hspace{6mm}d_{\mathcal{T},w}(p_\gamma',q_\gamma') &\leq w(p) + d_{\mathcal{T}}(p,r) + w(r) + w(r) + d_{\mathcal{T}}(r,q) + w(q)&&\nonumber\\
        &= d_{\mathcal{T},w}(p,r) + d_{\mathcal{T},w}(r,q)&&\nonumber\\
        &= d_{\mathcal{T},w}(p,q)\label{eq50}&&\\
    &\text{[since} \ \pi(p,q) \  \text{intersects} \ \gamma \ \text{at} \ r\text{].}&&\nonumber
\end{flalign}}

\noindent
Replacing $p_{\gamma}$ (resp. $q_{\gamma}$) by ${p_{\gamma}}'$ (resp. ${q_{\gamma}}'$) in inequality (\ref{eq44}),
{\setlength{\abovedisplayskip}{0pt}
\begin{flalign}
\hspace{6mm}d_{\mathcal{G} \setminus S'}(p,q) &\leq (4+5\epsilon).d_{\gamma,w}(p_\gamma',q_\gamma')&&\nonumber\\
	&\leq (4+5\epsilon) d_{\mathcal{T},w}(p,q) \hspace{0.2in} \text{[from} \ (\ref{eq50})\text{.]}&&\nonumber
\end{flalign}}
Thus $\mathcal{G}$ is a geodesic $(k, (4+\epsilon))$-VFTAWS for $S$. 
\end{proof}

\section{Conclusions}
\label{sect:conclu}

In this paper, we gave algorithms to achieve $k$ vertex fault-tolerance when the metric is additive weighted.
We devised algorithms to compute a $(k, 4+\epsilon)$-VFTAWS when the input points belong to any of the following: $\mathbb{R}^{d}$, simple polygon, polygonal domain, and the terrain.
Apart from the efficient computation, it would be interesting to explore the lower bounds on the number of edges for the fault-tolerant additive weighted spanners. 
Besides, the future work in the context of additive spanners could include finding the relation between the vertex-fault tolerance and the edge-fault tolerance, and optimizing various spanner parameters, like degree, diameter and weight.

\bibliographystyle{plain}


\end{document}